\documentclass[10pt]{article}
\usepackage{enumerate}
\usepackage{amsmath,amssymb}
\usepackage{natbib}
\usepackage{caption}
\usepackage{subcaption}
\usepackage[usenames]{color}
\usepackage{bm}
\usepackage{ying}
\usepackage{multirow}
\usepackage{rotating}
\usepackage{fullpage}
\usepackage{authblk}
\usepackage{enumerate}
\usepackage{geometry}
\usepackage{float}

\usepackage{mathtools}
\mathtoolsset{showonlyrefs}

\newcolumntype{g}{>{\columncolor{red}}c}

\usepackage{graphicx,amssymb}
\usepackage{xcolor}

\usepackage[colorlinks,
            linkcolor=red,
            anchorcolor=blue,
            citecolor=blue
            ]{hyperref}
\usepackage{algorithm}
\usepackage{algorithmic}

\let\hat\widehat
\let\tilde\widetilde

\def\asto{{\stackrel{\textrm{a.s.}}{\to}}}

\def \iid {\stackrel{\textnormal{i.i.d.}}{\sim}}

\def \calib {\textnormal{calib}}
\def \train {\textnormal{train}}
\def \test {\textnormal{test}}

\def \fdp {\textrm{FDP}}
\def \fdr {\textrm{FDR}}
\def \dtm {\textrm{dtm}}
\def \cfbh {{\texttt{cfBH}}}
\def \cfbhnull {{\texttt{cfBH0}}}

\usepackage{xcolor}
\usepackage{paralist}

\definecolor{ForestGreen}{RGB}{34,139,34}
 
\def\##1\#{\begin{align}#1\end{align}}
\def\$#1\${\begin{align*}#1\end{align*}}

\makeatletter
\newcommand{\printfnsymbol}[1]{%
  \textsuperscript{\@fnsymbol{#1}}%
}
\makeatother

\title{Selection by Prediction with Conformal p-values}
\author[1]{Ying Jin} 
\author[1,2]{Emmanuel J. Cand\`es}
\affil[1]{Department of Statistics, Stanford University} 
\affil[2]{Department of Mathematics, Stanford University}
\date{}

\begin{document}

\maketitle

\begin{abstract}
Decision making or scientific discovery pipelines such as 
job hiring and drug discovery often 
involve multiple stages: before any  
resource-intensive step, there is often 
an initial screening that  
uses predictions from a machine learning model to 
shortlist a few candidates from a large pool.  
We study screening procedures that aim to select 
candidates  whose unobserved outcomes 
exceed user-specified values.  
We develop a method 
that wraps around any prediction model 
to produce 
a subset of candidates 
while controlling the proportion of falsely selected units.  
Building upon the conformal inference framework, 
our method first constructs p-values  
that quantify the statistical evidence for large outcomes; 
it then determines the shortlist 
by comparing the p-values to a threshold introduced 
in the multiple testing literature.
In many cases, the procedure selects candidates  
whose predictions are above a data-dependent threshold.  
Our theoretical guarantee holds under mild exchangeability 
conditions on the samples, generalizing existing results 
on multiple conformal p-values. 
We demonstrate the empirical performance of our method 
via simulations, 
and apply it to  job hiring 
and drug discovery datasets.
 
\end{abstract}

\section{Introduction}

Decision making and scientific discovery are resource intensive tasks: 
human evaluation is needed before high-stakes decisions 
such as job hiring~\citep{shen2019deep} and disease diagnosis~\citep{etzioni2003case}; 
several rounds of expensive clinical trials are required before a 
drug can receive FDA approval~\citep{fdadrug}. 
Early on, we often hope to identify viable candidates from a very large pool---consider 
hundreds of applicants to a position or hundreds of thousands of 
potential compounds that may bind to the target. 
In such problems,  
machine learning prediction is useful for an initial screening step to 
shortlist a few candidates; in later, more costly stages, 
only these shortlisted candidates are
carefully investigated to confirm the interesting cases.

This paper concerns scenarios where outcomes 
taking on higher values are of interest. 
Formally, 
suppose we have access to 
a set of training data $\{(X_i,Y_i)\}_{i=1}^n$ 
and a set of test samples $\{X_{n+j}\}_{j=1}^m$ 
whose outcomes $\{Y_{n+j}\}_{j=1}^m$ are unobserved, 
all $(X,Y)\in \cX \times\cY$ pairs 
being i.i.d.~from some arbitrary and unknown 
distribution.\footnote{Later on, we will relax the i.i.d.~assumption to 
exchangeablity conditions.}     
Given some thresholds $\{c_j\}_{j=1}^m$, 
our goal is to find as many test units with $Y_{n+j}>c_j$ 
as possible, while ensuring 
the false discovery rate (FDR), 
the expected proportion of 
errors ($Y_{n+j}\leq c_j$) among 
all shortlisted candidates, is controlled.  
To be specific, letting $\cR\subseteq\{1,\dots,m\}$ be the 
selection set, we define FDR 
as the expectation of the false discovery proportion (FDP), so that 
\#\label{eq:def_fdr}
\fdr = \EE[\fdp],\quad \fdp=    \frac{\sum_{j=1}^m \ind\{j\in \cR, Y_{n+j}\leq c_j\}}{1\vee|\cR|} ,
\#
where we denote $a\vee b=\max\{a,b\}$ for any $a,b\in \RR$, and 
the expectation is over the randomness 
of all training data and all test samples.  
The FDR is a natural measure of type-I error 
for binary classification~\citep{hastie2009elements}. 
For regression problems with a continuous response, 
counting the error 
is reasonable if 
each selected candidate incurs a similar cost.  
We discuss below potential applications 
with binary or quantitative outcomes.   

\paragraph{Candidate screening.}
Companies 
are turning to
machine learning to support 
recruitment decisions~\citep{faliagka2012application,shehu2016adaptive}. 
Predictions 
using automatic resume screening~\citep{amdouni2010web,faliagka2014line}, 
semantic matching~\citep{mochol2007improving} 
or AI-assisted interviews 
are used to screen and select candidates from a large pool. 
Related tasks include talent sourcing, e.g., finding  
people who are likely to search for new opportunities, 
and candidate screening, i.e., 
selecting qualified applicants 
before further human evaluation~\citep{hirearticle}. 
{One may be interested in controlling the FDR~\eqref{eq:def_fdr} 
for resource efficiency: each shortlisted candidate {incurs} similar costs 
such as communication for talent sourcing 
and interviews before the hiring decisions.  
In candidate screening, controlling FDR 
ensures  
most of the costs are devoted to 
evaluating and ranking qualified candidates. 
Job recruitment also has fairness concerns; an alternative goal 
is to ensure that qualified candidates do not get screened out before human evaluation. 
To this end, one can flip the sign of the outcomes,
and~\eqref{eq:def_fdr} represents the proportion 
of qualified candidates in the filtered out ones.}

\paragraph{Drug discovery.}
Machine learning is playing a similar role 
in accelerating the 
drug discovery pipeline. 
Early stages of drug discovery 
aim at finding molecules or compounds---from 
a diverse library~\citep{szymanski2011adaptation} 
developed by institutions across the world~\citep{kim2021pubchem}---with 
strong effects such as high binding affinity to a specific target.   
The activity of drug candidates can be evaluated by  
high-throughput screening (HTS)~\citep{macarron2011impact}. However, the capacity of this approach is quite limited in practice, and 
it is generally infeasible to screen the whole library of readily synthesized compounds. 
Instead, virtual screening~\citep{huang2007drug} by 
machine learning models has enabled the 
automatic search of promising drugs. 
Often, a representative (ideally diverse) subset 
of the whole library is evaluated by HTS; 
machine learning models are then trained 
on these data  
to 
predict other candidates' activity based on their physical, geometric  
and chemical features~\citep{carracedo2021review,koutsoukas2017deep,vamathevan2019applications,dara2021machine} 
and select promising ones for further HTS and/or clinical trials. 
Given the cost of subsequent investigation, 
false positives in this process are a major concern~\citep{sink2010false}. 
Ensuring that a
sufficiently large proportion of resources is devoted to 
promising drugs is thus important for the efficiency of the whole pipeline. \\

\vspace{-0.5em}

In these two examples, the 
FDR quantifies a trade-off between the resources devoted to 
shortlisted candidates (the selection set) and the 
benefits from finding  interesting candidates (the true positives). 
This interpretation is similar to the justification of FDR 
in multiple testing~\citep{benjamini1995controlling,benjamini1997multiple,benjamini2001control}:  
when evaluating a large number of hypotheses, the FDR measures the proportion 
of ``false leads'' for follow-up confirmatory studies. 
However, 
in our \emph{prediction problem},   
the affinity of a new drug is inferred not  from 
the observations, but from other similar compounds, i.e., other drugs in the training data. 
This perspective also 
blurs the distinction between statistical inference and prediction; 
we will draw more connections between these sub-fields later.

The FDR may not necessarily be interpreted as 
a resource-efficiency measure. 
In the next two examples, 
controlling the FDR, which limits the error in inferring the direction of outcomes, 
is relevant to monitoring risk in healthcare 
and counterfactual inference. 

\paragraph{Healthcare.} 
With increasingly available 
patient data, 
machine learning is widely adopted to 
assist human decisions in healthcare.  
For example,  many works use machine learning prediction 
for large-scale early disease diagnosis~\citep{shen2019deep,richens2020improving} 
and patient risk prediction~\citep{rahimi2014risk,corey2018development,jalali2020deep}. 
Calibrating black-box models is important in such high-stakes problems. 
When it is more desired to limit false negatives 
than false positives,   
machine learning prediction might be used to filter out 
low-risk cases, leaving  other cases for careful human evaluation. 
It is sensible to control the proportion of 
high-risk cases among all filtered out samples.

\paragraph{Counterfactual inference.} 
In randomized experiments that run over a period of time, 
inferring whether the patients have benefited from the treatment option 
compared to an alternative might inform decisions such as early stopping of the trial 
for some patients. 
More generally, inferring the benefit of certain patients  
also provides evidence on treatment effect heterogeneity. 
This is a counterfactual inference problem~\citep{lei2021conformal,jin2023sensitivity} 
in which one could predict the counterfactuals, i.e., 
what would happen \emph{should} one takes an alternative option,  
by learning from the outcomes of  patients under that option, 
and then compare the prediction to the realized outcomes. 
In this case, the set of  those declared as having benefited 
from the treatment is informative 
if the FDR is controlled. 
\\

The generic task 
underlying these applications is 
to find a subset of candidates 
whose not-yet-observed outcomes are of interest
(e.g.,~qualification or high binding affinity to the target)  
from a potentially enormous pool of test samples. 
This is often achieved by thresholding their 
test scores---the model prediction on the test samples---from 
models  built on 
a set of training data that are assumed to be from the same distribution. 
However, controlling the error in the selected set 
is a nontrivial task.  

\subsection{Why calibrated predictive inference is insufficient}
\label{subsec:motivate}

We 
consider a binary example to fix ideas, so that $\cY=\{0,1\}$.  
Our goal is to find test samples 
with $Y_{n+j}=1$.
A natural starting point is 
to train a machine learning   model 
that predicts (classifies)
$Y$ given $X$,  
with the hope that test samples with
higher predicted values 
are more promising. 
To achieve valid prediction, 
one could calibrate the model~\citep{vovk2005algorithmic} 
to output a prediction set $\hat{C}_{1-\alpha}(X)$ 
taking the form $\varnothing$, 
$\{0\}$, $\{1\}$ or $\{0,1\}$ 
with the prescription that $\hat{C}_{1-\alpha}(X)$ 
must cover the outcome $Y$ with probability 
at least $1-\alpha$ for some user-specified $\alpha\in(0,1)$. 
The probability is 
averaged over the randomness in the test sample 
and the training process. 

However, a prediction set 
with marginal coverage guarantees  
is insufficient for selection. 
For instance, one 
might consider selecting all test samples $j$ with $\hat{C}_{1-\alpha}( X_{n+j} )=\{1\}$. 
The FDR of the selected set would then be below $\alpha$ 
if $\hat{C}_{1-\alpha}( X_{n+j} )$ covers 
$1$ with probability at least $1-\alpha$ \emph{for the selected units}, that is, 
conditional on selection. 
However, this is clearly a false statement because 
predictive inference only ensures
$(1-\alpha)$ coverage averaged over \emph{all} test samples. 
In fact, no matter how large we set the coverage $(1-\alpha)$, 
such a naive approach might still return a selection set  
that contains too many uninteresting candidates.

It might be best to preview our results on a real-world drug discovery dataset properly introduced and studied later in the paper (Section~\ref{subsubsec:hiv}). In short, the goal is to find promising drug candidates, among thousands of molecules, 
that are active ($Y=1$)
for the HIV target.
This dataset is highly imbalanced in the sense that only $3\%$  
of the drugs are active, as is often the case in studies on drug discovery.
Our main purpose is here to rapidly demonstrate that a straightforward application of conformal prediction methods, selecting those leads with $\hat{C}_{1-\alpha}(X_{n+j})=\{1\}$, 
results in over-confident predictions in a sense described below.

We use a deep learning model (this is introduced in Section~\ref{subsubsec:hiv}) to  construct conformity scores, and ultimately,  conformal prediction sets that are one-sided in the sense that they only take on three possible values: $\varnothing$, $\{1\}$ and $\{0,1\}$; 
 (see  Appendix~\ref{app:subsec_motivate} for
details).
The left panel of Figure~\ref{fig:motivate} shows 
the FDR of the naive approach as a function of the confidence level $1-\alpha\in\{0.99,0.98,\dots,0.70\}$, along 
with the marginal miscoverage of conformal prediction sets and 
the proportion of cases in the test set for which $\hat{C}_{1-\alpha}(X)=\{1\}$.

\begin{figure}[h]
    \centering
    \begin{subfigure}[t]{0.49\linewidth}
        \centering
        \includegraphics[width=0.95\linewidth]{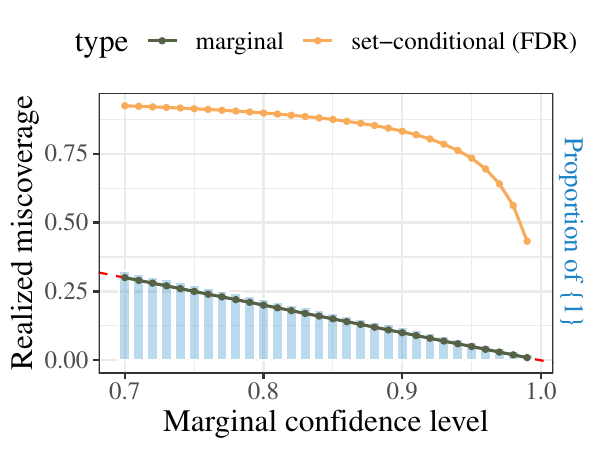}
    \end{subfigure} 
    \begin{subfigure}[t]{0.49\linewidth}
        \centering
        \includegraphics[width=0.95\linewidth]{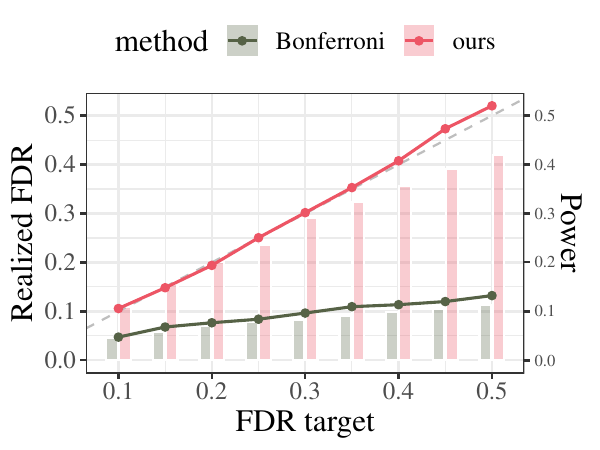}
    \end{subfigure} 
\caption{ 
Left: FDR (set-conditional miscoverage) of the naive approach  and marginal miscoverage as a function of the parameter $\alpha$; the light  blue bars are the  proportion  of cases among all test samples for which $\hat{C}_{1-\alpha}(X)=\{1\}$. Right: FDR (curve) and power (bar) of our selective inference approach and of  Bonferroni's method as a function of the nominal FDR target $q$. {The FDR (resp.~power) is computed  by averaging the FDP (resp.~proportion of true positives) in $N=100$ independent splits of training, calibration, and test data.}}
\label{fig:motivate}
\end{figure}

While conformal prediction always achieves nearly  exact 
marginal validity (brown), it is overconfident for seemingly promising candidates, as the error rate among the selected (FDR)--- those with $\hat{C}_{1-\alpha}(X)=\{1\}$---is very high (orange). When $1-\alpha = 0.90$, we witness an error rate of about 80\%, meaning that 4 out of 5 `discoveries' are false.  Even in 
the extremely conservative case where $1-\alpha=0.99$, the FDR exceeds 35\%. Note that this phenomenon is independent of the target FDR level. We can thus see that the selection issue 
would be especially pressing if, say, we aim for a small FDR level. 
In fact, conformal prediction  outputs a large proportion of uninformative sets: as seen from the light bars, about $1-\alpha$ of the prediction sets are $\hat{C}_{1-\alpha}(X)=\{0,1\}$ 
(we observe no empty prediction sets for this data). 
Thus, it ensures valid marginal coverage even though those $\hat{C}_{1-\alpha}(X)=\{1\}$ seldom 
cover the true label.

To make sure the FDR falls below a user specified tolerance $q\in\{0,1\}$, one might want to employ a Bonferroni correction. To do this we would pick test cases for which $\hat{C}_{1-q/m}(X_{n+j})=\{1\}$, where $m$ is the number of test samples. That is, we apply a Bonferroni correction to the marginal coverage, and this ensures that the probability of 
making a single false selection---which upper bounds the FDR---is below $q$. 
In the right panel of Figure~\ref{fig:motivate}, 
we compare the FDR and power of our approach and Bonferroni's method applied to a range of nominal FDR levels $q\in\{0.11,0.15,\dots,0.3\}$.\footnote{Here we take a subset of $m=1000$, 
as otherwise $q/m$ exceeds the resolution of conformal prediction.} 
Our approach yields almost exact FDR control and much higher power than Bonferroni's.

To ensure calibration on the selected, 
we will bridge conformal inference and selective inference 
and {devise \texttt{cfBH}, an algorithm that turns} 
any prediction model into a screening mechanism. 
In a nutshell, instead of calibrating to 
a fixed confidence level $\alpha$, 
we will use tools from conformal inference 
to quantify the model confidence in outcomes with larger values, 
and then employ multiple testing ideas to 
construct a shortlist of candidates with statistical guarantees.

Returning to the drug discovery application, 
we acknowledge a substantial literature 
using conformal inference for uncertainty quantification in
compound activity prediction, see \cite{lampa2018predicting,eklund2015application,svensson2018conformal,svensson2017improving,lindh2017predicting}, 
and \cite{cortes2019concepts} for a recent review. 
Whether explicitly stated or not, the goal is eventually to select  
or prioritize compounds that progress to 
later stages of drug discovery~\citep{ahlberg2017current,ahlberg2017using}  
after constructing valid prediction intervals. That said, current tools for selection  
are all heuristic, e.g.,  
picking cases with a high predicted value and a relatively short prediction interval. 
As already mentioned, a marginally valid 
prediction set does not necessarily imply reliable selection.  
The  method from this paper fills this gap, and can wrap around the predictions from the literature to  
produce reliable 
selection rules for drug discovery.

\subsection{Hypothesis testing and conformal p-values}

One 
may view our problem as testing the 
\emph{random}
hypotheses 
\#\label{eq:Hj}
H_j \colon Y_{n+j} \leq c_j,\quad j=1,\dots,m.
\#
From now on, we denote $\cH_0 = \{j\colon Y_{n+j}\leq c_j\}$ 
as the set of null hypotheses. 
That is, we define a hypothesis $H_j$ for 
each test sample $j$, and we say $H_j$ is non-null  
if $Y_{n+j}$ exceeds the threshold $c_j$. 
This is perhaps non-classical since 
the hypothesis $H_j$ is random: it 
concerns a random variable 
rather than a model parameter. 
However, we  
show that we can still use p-values  
and rely on multiple hypothesis testing ideas 
to construct the ``rejection'' set $\cR$.  

We start by introducing the tool 
we use to quantify model confidence: \emph{conformal p-values};
as its name suggests, these p-values  
build upon the 
conformal inference framework~\citep{vovk2005algorithmic,vovk1999machine}.  
Suppose we are given any prediction model 
from a training process that is 
independent of the calibration and test samples. 
We condition on the 
training process and view the prediction model as given. 
We first define a nonconformity score $V\colon \cX\times\cY\to \RR$ 
based on the prediction model.  
Intuitively,  $V(x,y)$ measures how well a value $y$ 
\emph{conforms} to the prediction of the model at $x$.  
For example, given a prediction $\hat\mu\colon \cX\to \RR$, 
one could use $V(x,y) = |y-\hat\mu(x)|$; 
other popular choices in the literature 
include 
ideas based on quantile regression~\citep{romano2019conformalized} 
and conditional density estimation~\citep{chernozhukov2021distributional}.  
Should $Y_{n+j}$ be observed, 
one could compute the nonconformity scores
$V_i=V(X_i,Y_i)$ for $i=1,\dots,n$ and $V_{n+j}=V(X_{n+j},Y_{n+j})$. 
The corresponding 
conformal p-value~\citep{vovk2005algorithmic,vovk1999machine,bates2021testing} 
is defined as 
\#\label{eq:pj*}
p_j^* = \frac{\sum_{i=1}^n \ind {\{V_i < {V}_{n+j} \}}+ U_j \cdot ( 1+ \sum_{i=1}^n \ind {\{V_i =  {V}_{n+j} \}})}{n+1}, 
\#
where $U_j\sim$ Unif[0,1] are i.i.d.~random variables 
to break ties. 
If the test sample $(X_{n+j},Y_{n+j})$ follows the same distribution 
as the training data, then $p_j^*\sim \text{Unif}\,[0,1]$.  
However, the mutual dependence among $\{p_j^*\}$  
is complicated 
as they all depend on the same calibration data.  
A recent paper
\citep{bates2021testing} used conformal p-values for 
outlier detection; in their setting,  
observations $\{(X_{n+j},Y_{n+j})\}_{j=1}^m$ are available, and 
the null set $\{j\colon H_j\text{ is true}\}$ 
is deterministic 
since the null hypothesis $H_j$ posits that $(X_{n+j},Y_{n+j})$ 
follows the same distribution as the training  samples. 
In our setting, 
the response $Y_{n+j}$ is not observed. 
This leads us to introduce a different set of conformal p-values.  
Our analysis also 
generalizes~\cite{bates2021testing} to 
exchangeable data.

\subsection{Related work}
\label{subsec:related_work}

This work concerns
calibrating prediction models  
to obtain correct directional conclusions on the outcomes.   
In situations where one cares more about the mistakes 
on the selected subset, our error notion, the FDR, might 
be more relevant than 
average prediction errors. 
That said, several works 
have studied FDR 
control in prediction problems, especially in 
binary classification. 
Among them,~\cite{dickhaus2014multiple} 
connects classification to multiple testing, 
showing that controlling type-I error (FDR) at certain levels 
by thresholding an oracle classifier 
asymptotically achieves the optimal (Bayes) classification risk;~\cite{scott2009false} 
provides
high-probability bounds for estimating the 
FDR achieved by classification rules, rather than adaptively 
controlling it at a specific level. 

Our problem setup is close to 
several recent works on calibrated screening 
or thresholding~\citep{wang2022improving,sahoo2021reliable} 
in classification or regression problems. 
These works however focus on different targets;~\cite{wang2022improving} 
focuses on selecting a subset with a prescribed expected number 
of qualified candidates;~\cite{sahoo2021reliable} focuses 
on the calibration of the predicted score itself 
to achieve a similar notion of error control as ours, but  
at varying levels for all thresholds. 
The difference is that our method rigorously controls FDR 
in finite samples, while it might be difficult to obtain 
such guarantees for 
the targets in~\cite{wang2022improving,sahoo2021reliable}.

Our methods build upon the 
conformal inference framework~\citep{vovk2005algorithmic,vovk1999machine}. 
Although 
conformal-inference-based methods  
have been developed for 
reliable uncertainty quantification 
in various problems~\citep{lei2021conformal,candes2021conformalized,jin2023sensitivity,tibshirani2019conformal}, 
the theoretical guarantee usually concerns 
a single test point. 
However, in many applications, 
one might be 
interested in a batch of individuals 
and desire uncertainty quantification 
for multiple test samples simultaneously; 
in such situations, these methods are insufficient 
due to the complex dependence structure of 
test scores and p-values as well as multiplicity issues.

This work is closely 
related to~\cite{bates2021testing},  
in which the authors use  conformal p-values~\eqref{eq:pj*} 
to test for multiple outliers. 
Our conformal p-values differ from theirs 
as the outcomes are not observed. 
A few works~\citep{mary2021semi,roquain2022false} 
are parallel to~\cite{bates2021testing}, studying 
multiple testing in a setting 
where all null hypotheses specify an identical null distribution; 
they are further generalized by~\cite{rava2021burden} to achieve 
subgroup FDR control in classification. 
Our method is similar to this line of work 
in constructing a threshold for certain ``scores''  
and selecting candidates with scores 
above that threshold. 
However, we work with random hypotheses and propose distinct procedures,  
whereas in their works, 
the hypotheses are deterministic (or conditioned on).  
We will discuss these distinctions in more detail  
as we present our results.  

Our perspective on the problem 
is also generally related to the multiple 
hypothesis testing literature where the FDR is a popular notion 
of type-I error.  
Since we pay more attention to one particular direction 
(e.g., we are interested in finding those $Y_{n+j}>c_j$), 
our work is related to  
testing the signs 
of  statistical parameters~\citep{bohrer1979multiple,bohrer1980optimal,hochberg1986multiple,guo2010controlling,weinstein2020online}.  
Our framework differs from 
the existing directional testing literature in important ways. 
Firstly, we test for 
the direction of a random outcome instead of a model parameter. 
This leads to random null sets, whose dependence structure is complicated. 
Secondly, 
our inference relies on exchangeability of the data 
while imposing no assumption on their distribution;  
this differs from 
standard practice, in which a null hypothesis 
specifies the distribution of test statistics 
and leads to a uniform p-value.

\section{Methodology}

\subsection{Selection by prediction with conformal p-values}
\label{subsec:basic}

We construct new conformal p-values 
to test the random hypotheses~\eqref{eq:Hj} regarding $Y_{n+j}$, 
building on an arbitrary 
nonconformity score $V$ obeying the
\emph{monotone} property.

\begin{definition}
A nonconformity score $V(\cdot,\cdot)\colon \cX\times\cY\to \RR$ is monotone if 
$V(x,y)\leq V(x,y')$ holds 
for any $x\in \cX$ and any $y, y'\in \cY$ obeying $y\leq y'$. 
\end{definition}

In general, $V(x,y)$ should measure how extreme the value $y$ 
is compared to the normal behavior of the outcome for 
a value $x$ of the covariate. 
For example, one could let $V(x,y) = y - \hat\mu(x)$ 
if the prediction machine outputs 
some estimate  $\hat\mu(x)$ of the conditional mean function 
or conditional quantile function. 
One could also set $V(x,y) $ as an estimator for $F(x,y):=\PP(Y\leq y\given X=x)$  
that obeys monotonicity~\citep{chernozhukov2021distributional}.   

The choice of $V$ may also account for 
the form of $\cR$:   
given a \emph{monotone} nonconformity score $V$,  
our method includes 
case $j$ in $\cR$ if $V(X_{n+j},c_j)$ is sufficiently small.  
As warm-up, 
consider binary classification. 
To find samples with $Y=1$, one could set $c_j\equiv 0.5$. 
Suppose $\hat\mu(\cdot)$ is some prediction from 
a machine learning algorithm; 
for instance, 
one could think of $\hat\mu(x)$ as an estimate for $\PP(Y=1\given X=x)$ 
obtained by means of a neural network.   
If one would like to select individuals 
with larger fitted probability $\hat\mu(X_{n+j})$, 
then $V(x,y)$ can be chosen in a way
such that $V(x,c_j)$ is decreasing in $\hat\mu(x)$. 
One such choice is 
$V(x,y) = y-\hat\mu(x)$.  
This reasoning also applies to continuous responses 
with regression modeling. 

We first compute $V_i = V(X_i,Y_i)$ for $i\in \cD_\calib=\{1,\dots,n\}$ 
and $\hat{V}_{n+j} = V(X_{n+j},c_j)$ for $j=1,\dots,m$; here  
we use the notation  $\hat{V}_{n+j}$ to distinguish from 
the unobserved scores $V_{n+j}:=V(X_{n+j},Y_{n+j})$. 
Then for each $j=1,\dots,m$, we construct the conformal p-values 
\#\label{eq:pj}
p_j = \frac{\sum_{i=1}^n \ind {\{V_i <\hat{V}_{n+j} \}}+  ( 1+ \sum_{i=1}^n \ind {\{V_i = \hat{V}_{n+j} \} } ) \cdot U_j}{n+1},
\#
where $U_j\sim \textrm{Unif}(0,1)$ are i.i.d.~random variables 
to break ties. 

\begin{remark} \normalfont Our conformal p-values have an intuitive
    interpretation: $p_j$ is the smallest significance level such that
    a one-sided conformal prediction interval for $Y_{n+j}$ excludes
    $c_j$.  Indeed, the split conformal inference
    procedure~\citep{lei2018distribution} (using $-V$ as the
    nonconformity score) yields the one-sided prediction interval
    $\hat{C}(X_{n+j},1-\alpha)=[\eta(X_{n+j},1-\alpha),+\infty)$ for
    $Y_{n+j}$, where \$ -\eta(X_{n+j},1-\alpha) =
    \textrm{Quantile}\big( {\textstyle 1-\alpha;
      \frac{1}{n+1}\sum_{i=1}^n \delta_{-V_i} +
      \frac{1}{n+1}\delta_{-\infty} } \big).  \$ It is guaranteed that
    $\PP(Y_{n+j}\in \hat{C}(X_{n+j},1-\alpha))\geq 1-\alpha$ where the
    expectation is taken over the randomness in $\cD_\calib$ and
    $\{X_{n+j},Y_{n+j}\}$.  Thus, ignoring the tie-breaking $U_j$, we see that the
    conformal p-value $p_j$  
    is the smallest $\alpha$ such that
    $c_j < \eta(X_{n+j},1-\alpha)$. Put it another way, we have
    confidence of at least $1-p_j$ that $Y_{n+j}>c_j$.
\end{remark}

The difference between $\{p_j\}$ in~\eqref{eq:pj} 
and $\{p_j^*\}$ in~\eqref{eq:pj*} is whether 
we use $\hat{V}_{n+j}$ or $V_{n+j}$ to contruct the p-values. 
To distinguish, we call $\{p_j^*\}$ the \emph{oracle} 
conformal p-values hereafter, 
which 
are not observable in our setting. 
In~\cite{bates2021testing}, 
$p_j^*$ quantifies how 
extreme a score is and is used to test 
whether the test sample $j$ is an outlier. 
In our context, $p_j$ quantifies how 
extreme the threshold is compared to 
the usual behavior of the outcomes. 




We then run the Benjamini-Hochberg (BH)
procedure~\citep{benjamini1995controlling,benjamini2001control} 
with the conformal p-values $\{p_j\}$. 
The whole 
procedure for \texttt{cfBH} is summarized in Algorithm~\ref{alg:bh}.

\begin{algorithm}[h]
  \caption{\texttt{cfBH}: Selection by prediction with conformal p-values}\label{alg:bh}
  \begin{algorithmic}[1]
  \REQUIRE Calibration data $\{(X_i,Y_i)\}_{i=1}^n$, 
  test data covariates $\{X_{n+j}\}_{j=1}^m$, 
  thresholds $\{c_j\}_{j=1}^m$, 
  FDR target $q\in(0,1)$, monotone nonconformity score $V\colon \cX\times\cY\to \RR$. 
  \vspace{0.05in} 
  \STATE Compute $V_i = V(X_i,Y_i)$ for $i=1,\dots,n$,  
  and $\hat{V}_{n+j}= V(X_{n+j},c_j)$ for $j=1,\dots,m$.
  \STATE Construct conformal p-values $\{p_j\}_{j=1}^m$ as in~\eqref{eq:pj}. 
  \STATE (BH procedure) Compute $k^* = \max\big\{k\colon \sum_{j=1}^m \ind\{p_j\leq qk/m\} \geq k\big\}$.
  \vspace{0.05in}
  \ENSURE Selection set $\cR = \{j\colon p_j \leq qk^*/m\}$.
  \end{algorithmic}
\end{algorithm}

\subsection{Finite sample FDR control}

Throughout, we define 
the random  vector $Z_i=(X_i,Y_i)$ for $1\leq j\leq n+m$, 
and $\tilde{Z}_{n+j} = (X_{n+j},c_j)$ for $1\leq j \leq m$. 
The following theorem 
establishes generic conditions under which 
\texttt{cfBH} controls  
the FDR~\eqref{eq:def_fdr} using $\{p_j\}$ in~\eqref{eq:pj}
with i.i.d.~calibration and test samples.

\begin{theorem}\label{thm:exchangeable}
Suppose $V$ is monotone, the calibration data $\{ Z_i \}_{i=1}^n$ and 
test data $\{ Z_{n+j} \}_{j=1}^m$ are i.i.d., 
and data in 
$\{Z_i\}_{i=1}^n \cup \{\tilde{Z}_{n+\ell}\}_{\ell\neq j} \cup \{Z_{n+j}\}$  
are mutually independent for any $j$. 
Then, for any $q\in (0,1)$, 
the output $\cR$ of Algorithm~\ref{alg:bh} satisfies 
$\mathrm{FDR}\leq q$. 
\end{theorem}

{Our framework applies to scenarios where $c_j$ are random variables 
(see Examples~\ref{ex:random},~\ref{ex:missing} and~\ref{ex:counterfactual} in the next subsection). 
In this case, all data being i.i.d.~does not necessarily imply the mutual independence of 
data in $\{Z_i\}_{i=1}^n \cup \{\tilde{Z}_{n+\ell}\}_{\ell\neq j} \cup \{Z_{n+j}\}$.  }

We note some conceptual novelty regarding  
our p-values in \texttt{cfBH}. 
In  conventional statistical inference,  
a null hypothesis (approximately) specifies the distribution 
of a test statistic, e.g., a p-value dominating Unif\,$[0,1]$. 
In sharp contrast, 
p-values defined in~\eqref{eq:pj} 
do not satisfy such a property, i.e., 
it does not necessarily hold that 
$\PP(p_j\leq \alpha \given j\in \cH_0)\leq \alpha$ for $\alpha\in(0,1)$. 
This is because $p_j$ and the random hypothesis $H_j$ are 
dependent in an unknown fashion.\footnote{{To see this, consider a special 
case where $V(x,y)=y-x$ 
for $Y=-X+\epsilon$ with $(X,\epsilon)\iid N(0,1)$, and $c=0$. 
With sufficiently many calibration data ($n\to \infty$), one can show that 
$p_j =\Phi(-X_{n+j})$ where $\Phi$ is the c.d.f.~of standard normal distribution. 
One can check that in this case, $\PP(p_j\leq 0.05\given Y\leq 0)>0.09$.}}
Instead, they obey
a generalized notion which states that 
\#\label{eq:general_typeI}
\PP(p_j \leq \alpha  \text{ and } j\in \cH_0 ) \leq \alpha,\quad \text{for all } \alpha\in[0,1].
\#
That is, for testing one single hypothesis $H_j$ at level $\alpha$, 
rejecting $p_j\leq\alpha$ yields the control 
of the generalized error~\eqref{eq:general_typeI} 
that accounts for the randomness in $H_j$ as well. 
This might connect to the Bayesian perspective where 
parameters are themselves random variables.

We outline some important properties 
of our p-values to develop intuitions regarding the FDR control of \cfbh. 
It is proved in
\cite{bates2021testing} that 
the oracle p-values 
$\{p_j^*\}$ satisfy a specific dependence structure called 
\emph{positive regression dependent on a subset} (PRDS)~\citep{benjamini2001control}.

\begin{definition}[PRDS]
  \label{def:prds}
  A random vector $X=(X_1,\dots,X_m)$ is PRDS 
  on a subset $\cI$ if for any $i\in \cI$ 
  and any increasing set $D$, 
  the probability $\PP(X\in D\given X_i=x)$ is increasing in $x$. 
\end{definition}
Here a set $D\subseteq \RR^m$ is 
\emph{increasing} if $a\in D$ and $b\succeq a$ implies $b\in D$, 
where $\succeq$ denote coordinate-wise inequality. 
To be specific,~\cite{bates2021testing} proved that 
the random vector of oracle conformal p-values $(p_1^*,\dots,p_m^*)$ is PRDS 
on  the index set $\cI$ consisting of 
all cases $(X_{n+j},Y_{n+j})$ that follow the same distribution as the training data. 
Relying on this result, 
we show that our p-values are PRDS after modifying one coordinate. 
The following lemma 
shows {the key properties for proving} Theorem~\ref{thm:exchangeable}. 
\begin{lemma}\label{lem:key}
  (i) If $V$ is monotone, then  
  $p_j \geq p_j^*$ on the event $\{j\in \cH_0\}$ for each $j$. 
  (ii) If the calibration and 
  test samples are i.i.d., then  $p_j^*\sim \textrm{Unif}\,([0,1])$. 
  (iii) If data in $\{Z_i\}_{i=1}^n \cup \{\tilde{Z}_{n+\ell}\}_{\ell\neq j} \cup \{Z_{n+j}\}$  
are independent,  
then $(p_1,\dots,p_{j-1},p_j^*,p_{j+1},\dots,p_{m})$ is PRDS on $p_j^*$.
\end{lemma}

The first  {two properties} in Lemma~\ref{lem:key} 
mean that $p_j$ is more conservative 
than $p_j^*\sim\text{Unif}\,[0,1]$ on the null event, hence 
leading to~\eqref{eq:general_typeI}. 
Generalizing to multiple hypotheses testing,   
one could expect that the false discoveries from 
$\{p_j\}$ 
can be controlled with those using $\{p_j^*\}$; 
the latter is studied in~\cite{bates2021testing} 
and works well with the Benjamini-Hochberg (BH)
procedure~\citep{benjamini1995controlling,benjamini2001control} 
because of the PRDS property (c.f.~Definition~\ref{def:prds}). 
These heuristics are made concrete by 
a careful leave-one-out 
analysis for FDR control, 
as well as the general 
PRDS property of conformal p-values in (iii) of Lemma~\ref{lem:key}. 
We defer the detailed  proofs to Appendix~\ref{app:subsec_bh}.

\subsection{Extension to exchangeable data}

Our previous result extends to the situation where the calibration and test samples obey a natural exchangeability condition. 
\begin{theorem}\label{thm:determ}
  Suppose $V$ is monotone, and that for any $j=1,\dots,m$, 
  the random variables $\{V_1,\dots,V_n, V_{n+j}\}$ are exchangeable 
  conditional on $\{\hat{V}_{n+\ell}\}_{\ell\neq j}$. 
  Also,  the random variables $\{V_i\}_{i=1}^n\cup\{\hat{V}_{n+\ell}\}_{\ell=1}^m$ 
  have no ties almost surely. 
  Then $\cfbh$ applied to p-values defined as 
  \#\label{eq:pj_determ}
  p_j^{\textnormal{\dtm}} = \frac{1 + \sum_{i=1}^n \ind {\{V_i <\hat{V}_{n+j} \}} }{n+1},
  \#  
  satisfies
  $\textnormal{FDR}\leq q$.
  \end{theorem}

 This result may be of interest in the case where one is sampling from a finite set without replacement.
 {For instance, in drug discovery, the 
 calibration data may be molecules that have already been evaluated; 
 to curate such dataset, 
 it is common to randomly sample and evluate a fixed number of molecules 
 from a fixed drug library.  
 Conditional on all the data in the  library, 
 the randomness from sampling still ensures 
 the exchangeability conditions in Theorem~\ref{thm:determ}. 
 On the contrary, 
 the i.i.d.~assumptions in~Theorem~\ref{thm:exchangeable} may not hold 
 under sampling without replacement.}
The proof of Theorem~\ref{thm:determ} is in Appendix~\ref{app:determ}. 
Its analysis no longer relies on 
the PRDS property developed in~\cite{bates2021testing} for i.i.d.~data.\footnote{A consequence is that 
our new technique can be used to 
show finite sample FDR control for the outlier detection 
problem in~\cite{bates2021testing}  
under a similar exchangeability condition.}

\subsection{Setting the testing thresholds} 

Our framework allows us to test the random hypothesis~\eqref{eq:Hj} 
for general thresholds $\{c_j\}$.  
The simplest case is to set $c_j=\tau$, where 
$\tau$ is some constant which could possibly be obtained from 
an independent training process.  
It covers binary or one-versus-all classification problems 
as well as many scenarios in regression modeling 
where a threshold on the outcomes can be decided beforehand. 
Consider an application 
of large-scale screening in early-stage diagnosis~\citep{shen2019deep}, 
where the practitioner would like to find individuals 
with high unobserved health risk.
The threshold of the risk measure as being hazardous 
could be decided by domain knowledge, 
or by looking at the
experience of former patients.  
In this case, as  $\tau$ is  independent of all subsequent steps, it  
can be viewed as fixed and 
the conditions in Theorem~\ref{thm:exchangeable} are thus satisfied. 
However, we do note that $\tau$ should not depend on 
the calibration data, because this would potentially break the 
mutual independence condition 
on $\{Z_i\}_{i=1}^n$ and $\{\tilde{Z}_{n+\ell}\}_{\ell\neq j}$
in Theorem~\ref{thm:exchangeable} and invalidate  FDR control. 

More generally,  
$c_j$ can also be a random variable  for test sample $j$ as discussed below. 

\vspace{-0.5em}

\begin{example}[Random variable associated with test sample]\normalfont
\label{ex:random}
In early disease diagnosis, 
one might want to identify individuals 
whose future cholestrol level $Y_{n+j}$ in a month 
would be higher than $c_j := W_{n+j}$, their current measurements 
upon a first visit. 
In this case, we simply modify 
the construction of $p_j$ in Algorithm~\ref{alg:bh} 
to $\hat{V}_{n+j} = V(X_{n+j},W_{n+j})$.  
All the conditions in Theorem~\ref{thm:exchangeable} remain true 
as long as $\{(X_{n+j},Y_{n+j})\}_{j=1}^m$ are i.i.d.~pairs, 
and 
$\{(X_{n+j},Y_{n+j},W_{n+j})\}_{j=1}^m$ are i.i.d.~triples, hence
our method still yields valid FDR control.  
\end{example}

Continuing the above example,  
we allow for situations where 
the random variable 
$W$ is unobserved (missing)
in the calibration data. 

\vspace{-0.5em}

\begin{example}[Missing variable in the calibration data] \normalfont
\label{ex:missing}
While the goal is to compare the future cholesterol level  
of a new patient to that upon admission,  
the cholesterol level upon a first visit, $W_i$,  may not be measured 
for former patients $i=1,\dots,n$ in the calibration set. 
This setting cannot be turned into a 
classification problem because the $(W,Y)$ pair is 
never simultaneously observed for calibration samples. 
However, our method is still applicable with $\hat{V}_{n+j}=V(X_{n+j},W_{n+j})$ 
as the conditions in Theorem~\ref{thm:exchangeable} still hold.  
\end{example}

Another example that is closely related to the missing variable setting 
is to infer whether the counterfactual is larger or smaller 
than the realized outcome in causal inference. 

\vspace{-0.5em}

\begin{example}[Counterfactual inference] \normalfont
\label{ex:counterfactual}
Predicting counterfactuals 
is another application whose setup is similar to Example~\ref{ex:missing}~\citep{lei2021conformal,jin2023sensitivity}. 
Under the potential outcomes framework~\citep{imbens2015causal}, we 
let $\{X_i,T_i,O_i(1),O_i(0)\}_{i=1}^N$ be i.i.d.~tuples from an unknown 
super-population $\PP$, where $X_i\in \cX$ is the vector of covariates, 
$T_i\in\{0,1\}$ is the treatment indicator, and $(O_i(1),O_i(0))\in \RR^2$ are 
potential outcomes under treatment and control, respectively. 
We observe  
$\{(X_i,O_i,T_i)\}_{i=1}^N$ where $O_i=O_i(T_i)$.  
We consider completely randomized experiments  
where for some $p\in(0,1)$, 
$T_i\sim \textrm{Bern}(p)$ are independent of all other quantities. 
Counterfactual inference (e.g.~predicting
$O_i(0)$ when $T_i=1$) asks what would happen should 
unit $i$ receive  another treatment status. 
A potential application is to find treated units with $O_{i}(1)\leq O_{i}(0)$ 
so that they might drop out early from the experiment to avoid adverse effects. 
In this case, we take all the control units $\{(X_i,O_i(0))\}_{i=1}^n$ as the calibration data,  
each i.i.d.~from $\PP_{X,O(0)\given T=0}$. 
The test data are $\{X_{n+j}\}_{j=1}^m$  for which 
$\{O_{n+j}(0)\}_{j=1}^m$ are not observed.  
That is, we set the response as $Y=O(0)$ for all samples
and $c_j = O_{n+j}(1)$ (the observed outcome) in~\eqref{eq:Hj} which is a random variable.  
This task cannot be turned into a 
classification problem because  $(O_i(1),O_{i}(0))$  is 
never simultaneously observed for any unit. 
However, letting $\hat{V}_{n+j}=V(X_{n+j},c_j)$, 
the conditions in Theorem~\ref{thm:exchangeable} still hold, because 
randomized treatments imply
$\PP_{X,Y\given T=1}=\PP_{X,O(0)}=\PP_{X,Y\given T=0}$.  
In this way, our method identifies multiple 
treated units among whom a prescribed proportion have
negative individual treatment effects. 
\end{example}

\subsection{Asymptotic analysis and choice of nonconformity score}
\label{subsec:asymp}

While the only requirement for the validity of $\cfbh$ 
is the monotonicity of the nonconformity score function, 
a carefully constructed score might 
enhance the 
power. 
This concerns two aspects: 
(i) what should the prediction machine pursue (as a function of $x$), 
and (ii) given any output of 
the prediction machine, which nonconformity score (as a function 
of both $x$ and $y$) should 
the practitioner use. 
We offer some heuristics  
through asymptotic analysis. 

To simplify the discussion, we take the thresholds as $c_j=0$, 
and assume the samples $\{(X_i,Y_i)\}_{i=1}^{n+m}$ are i.i.d.~from 
an unknown distribution. 
In this case, the outcomes can be encoded into a binary variable 
$\ind \{Y_{i}>0\}$ for $i\in \cD_\calib\cup\cD_\test$. 
We thus consider the case $Y\in\{0,1\}$, 
and define 
\$
\text{Power} = \EE\Bigg[ \frac{\sum_{j=1}^m \ind\{j\in \cR, Y_{n+j}=1\}}{1\vee\sum_{j=1}^m \ind\{Y_{n+j}=1\}}  \Bigg].
\$
We consider the regime where both $n$, the size of the calibration set, 
and $m$, the size of the test set, 
tend to infinity. 
Throughout, 
we hold the nonconformity score $V$ as fixed, 
and the randomness is only in the calibration 
and test samples.  
The proof of 
the next proposition is in Appendix~\ref{app:prop_asymp}. 

\begin{prop}
\label{prop:asymptotic}
Let $V$ be any fixed monotone nonconformity score, and 
suppose $\{(X_i,Y_i)\}_{i=1}^{n+m}$ are i.i.d. 
Define $F(v,u) = \PP(V(X,Y)<v) +u\cdot \PP(V(X,Y)=v) $ 
for any $v\in \RR$ and $u\in[0,1]$. 
Define 
$
t^* =  \sup\big\{ t\in [0,1]\colon  t/\PP ( F(V(X,0), U ) \leq t  )  \leq q  \big\}. 
$
Suppose that for any sufficiently small $\epsilon>0$, 
there exists some $t\in(t^*-\epsilon,t^*)$ such that $t/\PP ( F(V(X,0), U ) \leq t  )  < q$. 
Then the output $\cR$ of $\cfbh$ satisfies 
\$\normalfont
&\lim_{n,m\to \infty} \fdr = \frac{ \PP \{ F(V(X,0),U)\leq t^*, Y \leq 0  \}}{\PP\{F(V(X,0),U)\leq t^* \} },\quad \textrm{and} \\
&\lim_{n,m\to \infty} \textnormal{Power} = 
\frac{ \PP \{ F(V(X,0),U)\leq t^*, Y>0 \}}{\PP\{Y>0 \} }.
\$
\end{prop}

In words, Proposition~\ref{prop:asymptotic} states that   
the rejection threshold for conformal 
p-values in $\cfbh$ converges 
to $t^*$, leading to the convergence of FDR. 
{By definition, $t^*$ is the largest value of $t$ such that  
the mass of $F(V(X,Y),U)$ below $t$ is less than $q$ times 
that of $F(V(X,0),U)$. 
Since $F(v,u)$ is increasing in the first argument,  
$t^*$ is large  if $V(X,0)$ has a far more heavier left tail than $V(X,Y)$. 
One such case is when the probability of $Y=1$ is large for small $V(X,0)$. 
For instance, if we use $V(x,y):=y-\hat\mu(x)$ for a point prediction $\hat\mu(\cdot)$, 
then $t^*$ is large if $Y=1$ is more probable for $X=x$ with large $\hat\mu(x)$.}  
The convergence is not necessarily true if 
the asymptotic FDR lies on the critical line near $t^*$, i.e., if 
there exists some $t_1<t^*$ such that 
$t/\PP ( F(V(X,0), U) \leq t  ) = q$ for all $t\in[t_1,t]$. 
The technical condition in Proposition~\ref{prop:asymptotic} 
rules out this situation, and actually guarantees that 
\$
q = \frac{t^* }{ \PP \{ F(V(X,0), U) \leq t^* \} } =  \frac{\PP \{ F(V(X,Y), U) \leq t^* \} }{ \PP \{ F(V(X,0), U) \leq t^*  \} } 
\$ 
since the distribution of $F(V(X,0),U)$ has no point mass. 
In particular, 
the asymptotic FDR of $\cfbh$
satisfies 
\#\label{eq:bd_fdr_asymp}
\frac{ \PP \{ F(V(X,0),U)\leq t^*, Y \leq 0  \}}{\PP\{F(V(X,0),U)\leq t^* \} }
&\leq 
\frac{ \PP \{ F(V(X,Y),U)\leq t^*, Y \leq 0  \}}{\PP\{F(V(X,0),U)\leq t^* \} }  \notag \\
&\leq \frac{ \PP \{ F(V(X,Y),U)\leq t^*   \}}{\PP\{F(V(X,0),U)\leq t^* \} } \leq q.
\#
 
To further simplify, we suppose 
the distribution of $V(X,Y)$ does not have a
point mass, hence $F(u,v) = \PP(V(X,Y)\leq v)$. 
We also let $v^* = \sup\{v\colon \PP(V(X,Y)\leq v)\leq t^*)$, 
such that $F(u,v)\leq t^*$ if and only if $v\leq v^*$. 
Thus,~\eqref{eq:bd_fdr_asymp} becomes 
\#\label{eq:bd_fdr_asymp_01}
\frac{ \PP \{ V(X,0)\leq v^*, Y = 0  \}}{\PP\{V(X,0)\leq s^* \} } = 
\frac{ \PP \{ V(X,Y)\leq v^*, Y = 0  \}}{\PP\{V(X,0)\leq s^* \} }  
\leq \frac{ \PP \{ V(X,Y)\leq v^*  \}}{\PP\{ V(X,0)\leq v^* \} } \leq q.
\#

\paragraph{Choice of $V$.} 

We first investigate~\eqref{eq:bd_fdr_asymp_01} 
to provide some heuristics on the choice of $V$. 
If 
we could design some nonconformity score $V$ such that 
\#\label{eq:iff}
V(X,Y)\leq v^* \quad \Rightarrow \quad Y=0,
\#
then the first inequality in~\eqref{eq:bd_fdr_asymp_01} becomes an equality.
Given a prediction $\hat\mu(x)$ from any machine learning algorithm, 
if one would like to select individuals with larger values of $\hat\mu(X_{n+j})$, 
one might design a nonconformity score $V$ such that 
\$
V(x,0) = - \hat\mu(x),\quad V(x,1) = +\infty. 
\$
In this way, selecting cases where $V(X_{n+j},0)$ 
is small is equivalent to selecting large $\hat\mu(X_{n+j})$, 
and this choice guarantees~\eqref{eq:iff} as long as $v^*<\infty$. 
We recommend a relaxation given by 
\#\label{eq:clip}
V(x,y) = M\cdot y - \hat\mu(x)
\#
for some sufficiently large constant $M$.  
This  ``clipped'' score obeys
$\inf_{x} V(x,1) = M-\sup_x \hat\mu(x) \geq  \sup_{x}V(x,0)$ 
if $M\geq 2\sup_{x}|\hat\mu(x)|$. 
That is, the nonconformity score for $Y=1$ is always 
larger than that for $Y=0$, regardless of the value of $x$. 
Recalling the definitions, we know $t^* \leq q\cdot \PP(F(V(X,0),U)\leq t^*)\leq q$. 
Thus, when $q<\PP(Y=0)$, 
by definition, $v^*$ is smaller than the $q$-th quantile of $V(X,Y)$. 
As a result,~\eqref{eq:iff} holds exactly and the 
first inquality in~\eqref{eq:bd_fdr_asymp_01} is an equality---that is, 
using~\eqref{eq:clip} could potentially yield 
a value of FDR close to the nominal level $q$, using up all the FDR budget; 
we thus anticipate a higer power. 
We indeed verify these heuristics in our simulations.

\paragraph{Choice of $\hat\mu$.}
We then discuss the choice of $\hat\mu$ when 
$V(x,y) = My-\hat\mu(x)$ and~\eqref{eq:iff} holds. 
Recall that given the conditions in Proposition~\ref{prop:asymptotic}, 
the last inequalities in~\eqref{eq:bd_fdr_asymp} and~\eqref{eq:bd_fdr_asymp_01}
are exact equalities. Hence  
\$
\lim_{n,m\to\infty} \fdr = \frac{\PP\{-\hat\mu(X)\leq v^*,Y=0\}}{\PP\{-\hat\mu(X)\leq v^*\}}, \quad 
\lim_{n,m\to \infty} \text{Power} = 
\frac{ \PP \{ -\hat\mu(X)\leq v^*, Y=1 \}}{\PP\{Y=1 \} }.
\$ 
Since our procedure always ensures 
that the asymptotic FDR is below $q$, 
letting $f(x)=v^*+\hat\mu(x)$, 
we could view asymptotic power maximization as solving 
an optimization problem 
\$
\textrm{maximize} &\quad  \PP \{ f(X)\geq 0, Y=1  \} \\ 
\text{subject to}&\quad  \frac{\PP\{f(X)\geq 0,Y=0\}}{\PP\{f(X)\geq 0\}} \leq q. 
\$
Equivalently, this is 
\$
\textrm{maximize} &\quad  \EE\big[ \ind\{f(X)\geq 0\}  \PP(Y=1\given X)   \big] \\ 
\text{subject to}&\quad  \EE\big[ \ind\{f(X)\geq 0\} \big( \PP(Y=0\given X)-q\big)   \big] \leq 0. 
\$
By Neyman-Pearson lemma, 
the optimal choice of $f$ should be a monotone function of 
$
\PP(Y=1\given X)  
$. 
That is, we should 
aim for some $\hat\mu(x)$ that is monotone in $\PP(Y=1\given X=x)$. 
The most convenient option is to fit 
$\PP(Y=1\given X=x)$. 
This heuristic derivation leads to a quite 
intuitive recommendation: 
the predicted score should indeed aim to 
reflect how likely $Y=1$ is given $X$.

\section{Numerical experiments}
\label{sec:simu}

We evaluate our method 
on simulated datasets, 
leading to some practical suggestions.  
We generate i.i.d.~covariates $X_i\sim$Unif $[-1,1]^{20}$ 
and responses $Y_i = \mu(X_i)+ \epsilon_i$, where 
$\mu(x)=\mathbb{E}[Y\,|\, X=x]$ is nonlinear in $x$, 
and $\epsilon_i$ is the independent random noise.  
We design 8 simulation settings 
to demonstrate the performance of our methods 
under various data generating processes, 
with 
different configurations of  $\mu(\cdot)$ and 
distributions for the $\epsilon_i$'s. 
In particular, we vary 
(i) whether the image $\{\mu(x)\colon x\in[-1,1]^{20}\}$ 
is a continuous set, and 
(ii) whether the noise is heterogeneous, 
such that 
the hardness of correctly identifying those outcomes exceeding zero
varies. 
The details of
all settings 
are summarized in Appendix~\ref{app:subsec_simu_detail}. 
The reproduction codes for this part can be found at 
\url{https://github.com/ying531/selcf_paper}.

The task is to select individuals with $Y_{n+j}>0$ 
among all test samples. 
We fix the sizes of training and calibration data 
at $n = |\cD_\train| = |\cD_\calib| = 1000$ 
and vary the test sample size $|\cD_\test| \in \{10,100,500,1000\}$. 
We use gradient boosting, SVM with \texttt{rbf} kernel, 
and random forest to fit 
a regression model $\hat{\mu}(\cdot)$ for $\EE[Y\given X]$, 
all from the \texttt{scikit-learn} Python library without fine tuning. 
We then apply $\cfbh$ and Algorithm~\ref{alg:bh0} at the 
FDR target $q =0.1$, which, together with the Bonferroni baseline, 
leads to four algorithm configurations: 
\vspace{-0.3em}
\begin{enumerate}[itemsep=-0.5ex]
  \item \texttt{BH\_sub}: $\cfbhnull$ (Algorithm~\ref{alg:bh0}) with $\hat{V}_{n+j} = -\hat{\mu}(X_{n+j})$;
  \item \texttt{BH\_res}: $\cfbh$ with $V(x,y) = y - \hat\mu(x)$;
  \item \texttt{BH\_clip}: $\cfbh$ with $V(x,y) = M\cdot\ind\{y>0\} - \hat\mu(x)$ and  
  a large constant $M=100$; this value is chosen to 
  ensure it is larger than $ 2\sup_x|\hat\mu(x)|$; 
  \item \texttt{Bonferroni}: Select all $p_j \leq q/m$ with $V(x,y)$ the same as in \texttt{BH\_clip}.
\end{enumerate}  

{Algorithm~\ref{alg:bh0} used in \texttt{BH\_sub} is formally introduced in 
Appendix~\ref{app:condition_variant}; when applied to 
classification problems, it is equivalent to 
the score-based methods of~\cite{mary2021semi} and~\cite{rava2021burden}. 
The only difference from $\cfbh$ is that 
Algorithm~\ref{alg:bh0} ($\cfbhnull$) uses 
$\{(X_i,Y_i)\colon i\in \cD_\calib,~ Y_{i}=0\}$ as the calibration data  
when constructing conformal p-values~\eqref{eq:pj}, which  leads to a 
slightly stronger theoretical guarantee although it comes at a price: loss of power (see Appendix~\ref{app:subsec_classify}). 
Other than this, we are not aware of alternative methods for exact control of FDR in classification.}
\vspace{-0.3em} 

\subsection{Valid FDR control}

We empirically evaluate the FDR by averaging the FDP 
$
\frac{ \sum_{j\in \cD_\test}\ind\{j\in \cR,Y_{n+j}>0\}}{1 \vee |\cR|}
$
over $N=1000$ independent runs ($\cR$ is the rejection set). 
We observe similar power and FDR for different values of $n_\test$, 
hence we only plot 
the results for $q=0.1$ and $n_\test=100$ in Figure~\ref{fig:fdr_01}.  
The FDR is controlled below $q=0.1$ 
in all configurations, showing the validity 
of our procedure. 
{In particular, the FDR of \texttt{Bonferroni} is always close to zero.}

\begin{figure}[h]
  \centering
  \includegraphics[width=6in]{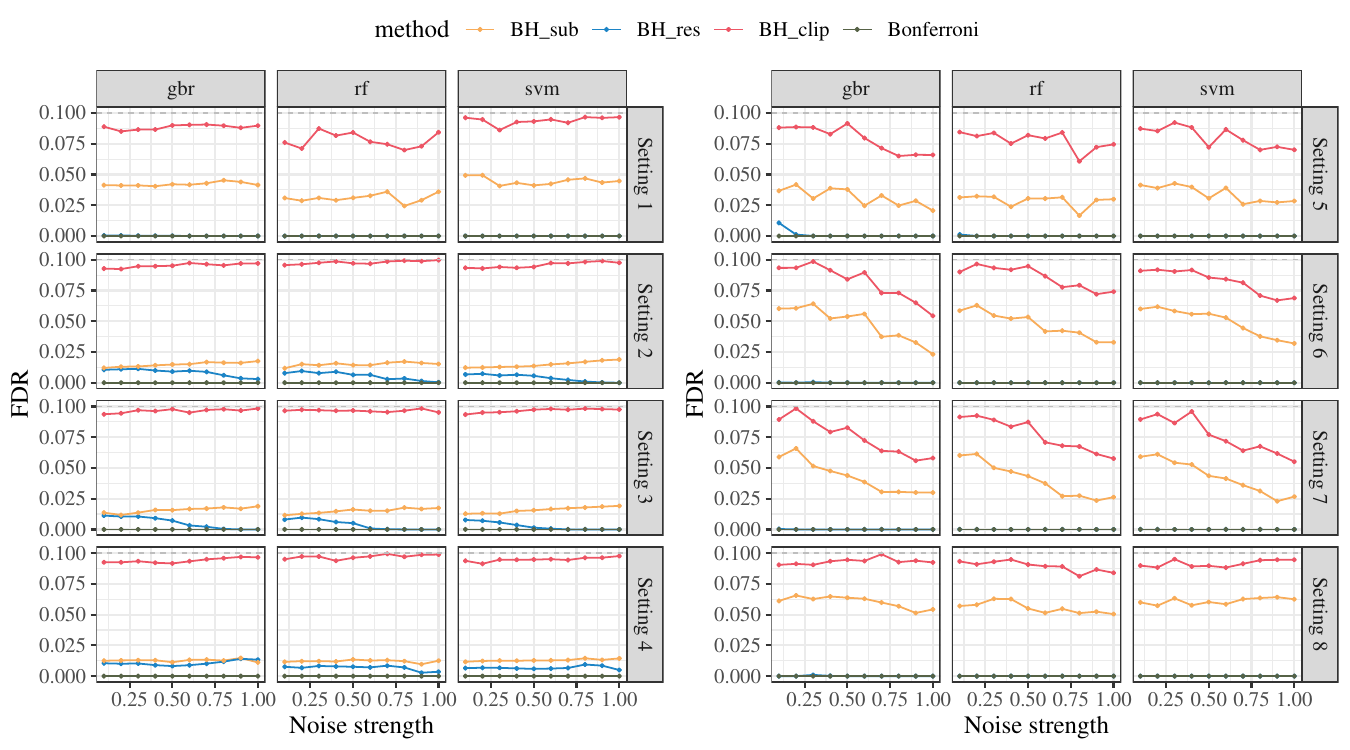}
  \caption{Realized FDR for four procedures at 
  FDR target $q=0.1$. 
  Each row corresponds to one data generating process, 
  and each column corresponds 
  to one regressor (\texttt{gbr} for gradient boosting, \texttt{rf} 
  for random forest, \texttt{svm} for support vector machine). 
  The $x$-axis is the parameter $\sigma$ for the noise level 
  of $\epsilon_i$, whose precise definition 
  is in Appendix~\ref{app:subsec_simu_detail}.}
  \vspace{-1em}
  \label{fig:fdr_01}
\end{figure}

Among the three nonconformity scores, 
the realized FDR of \texttt{BH\_clip} is the highest
in all settings, and is often very close to the nominal level. 
FDR also varies with the regression algorithms that are adopted, 
but the variation is not very large.

In settings 1-4 and 8, the FDR levels of different methods 
are relatively stable across various noise strengths. 
In settings 5-7, the FDR decreases 
as the noise level (the $x$-axis) increases.   
It might seem counterintuitive because 
at first sight, one might think that a harder problem  
(larger noise) would lead to a higher error rate. 
However, we observe that it is accompanied by  
lower power (Figure~\ref{fig:power_01} 
as we will present shortly) and a smaller rejection set (Figure~\ref{fig:nrej_01}  
in Appendix~\ref{app:subsec_plots}).  
This might be contributed by two factors. 
The first is the increased difficulty of prediction; with larger noise, 
machine learning  
is less capable of capturing the heterogeneity in the true 
conditional mean function $\mu(X_{n+j})$. 
Since 
a test sample needs to have a sufficiently small 
value of $V(X_{n+j},0)$ 
to be selected, such lack of heterogeneity 
leads to small selection sets. 
The second is the decreased confidence even with ground truth available: 
even when $\mu(X_{n+j})$ is known, 
when the noise is too large, 
there is hardly any value of the covariate for which 
one has a high confidence that $Y>0$. 
To be more specific, 
when $c_j=\tau$ is a constant, 
the selection set 
is fully decided by $\cD_\calib \cup\{X_{n+j}\}_{j\in \cD_\test}$; 
in addition, 
 $\{Y_{n+j}\}_{j\in \cD_\test}$ are independent of $\cD_\calib$ conditional on 
$\{X_{n+j}\}_{j\in \cD_\test}$ 
by the i.i.d.~assumption. 
The tower property then implies
\#\label{fdr:cond}
\fdr = \EE\Bigg[  \frac{\sum_{j=1}^m \ind\{j\in \cR\} \ind\{ Y_{n+j}\leq 0\}}{1\vee|\cR|}    \Bigg] 
= \EE\Bigg[  \frac{\sum_{j=1}^m \ind\{j\in \cR\} \PP\big( Y_{n+j}\leq 0\given X_{n+j}\big) }{1\vee|\cR|}    \Bigg],
\#
which is roughly the average of $\PP(Y_{n+j}\leq 1\given X_{n+j})$ among 
selected individuals. 
Thus, when $\PP(Y \leq 0 \given X=x) > q$ for almost all $x\in \cX$, 
sometimes one needs to output $\cR=\varnothing$ 
in order to keep the FDR below $q$, 
leading to smaller selection sets. 
In general, as the selection becomes difficult, 
selected units should have extremely small nonconformity scores 
and extremely strong confidence in a positive response, 
resulting in a lower FDR.

\subsection{Power}

We evaluate 
power by averaging 
$
\frac{ \sum_{j\in \cD_\test}\ind\{j\in \cR,Y_{n+j}>0\}}{\sum_{j\in \cD_\test}\ind\{Y_{n+j}>0\}},
$ 
the proportion of correct selections 
among all positive test samples, over all replicates. 
We again observe stable power across 
different values of $n_\test$, 
hence we only plot the average power for $q=0.1$ and $n_\test=100$ 
in Figure~\ref{fig:power_01}. 
\texttt{BH\_clip} always has the highest power 
while \texttt{BH\_res} always has the lowest power (excluding \texttt{Bonferroni}); 
\texttt{BH\_sub} is sometimes closer to \texttt{BH\_clip} 
and sometimes closer to \texttt{BH\_res}. 
We note that the general applicability of \texttt{BH\_res} 
comes with its low power in such binary classification problems, 
while the other two 
(which are only applicable for fixed thresholds) are more powerful.   
{Finally, the Bonferroni correction nearly has no power (even if we set 
$V(x,y)$ to be the same as the most powerful \texttt{BH\_clip}).}

\begin{figure}[h]
  \centering
  \includegraphics[width=6in]{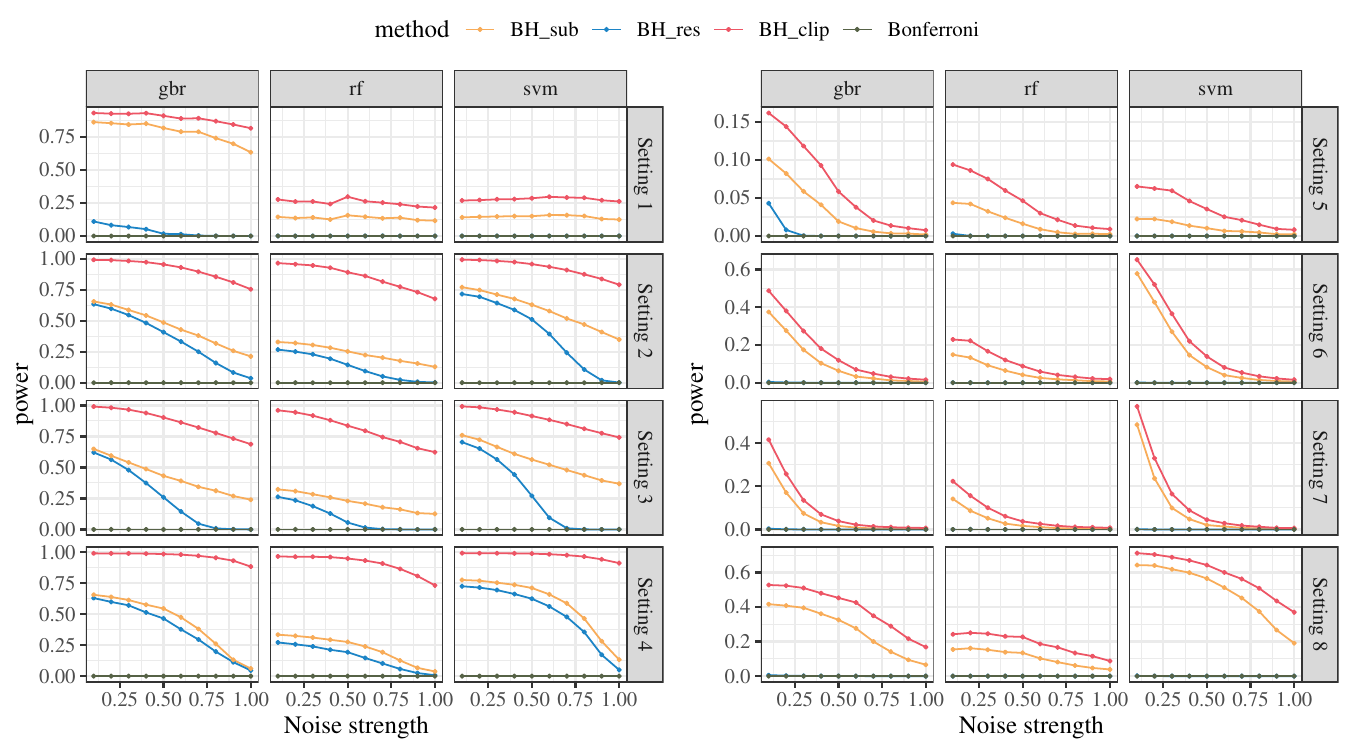}
  \caption{Realized power of four procedures at 
  FDR target $q=0.1$ for various data generating processes. 
  Details  are otherwise the same as Figure~\ref{fig:fdr_01}.}
  \label{fig:power_01}
  \vspace{-1em}
\end{figure}

The power of our procedure also varies with the prediction algorithms (the columns). 
In setting 1 where the FDR is similar across the three algorithms, 
the power is actually drastically different: \texttt{gbr} performs 
the best in most settings; 
however, \texttt{svm} performs the best in settings 4 and 8 
where there is strong heterogeneity in the distribution of $\epsilon_i\given X_i$, 
in which case the other two prediction algorithms might fail 
to capture the true dependence between $Y$ and $X$.

The power 
decreases as the noise strength increases in all settings. This is because 
larger noise makes it more difficult to 
fit the prediction model, and the fundamental detection hardness 
increases 
as sketched in~\eqref{fdr:cond}. 
\eqref{fdr:cond} also 
implies that in practice, the FDR target $q$ should 
be properly chosen: in situations where $\PP(Y_{n+j}>0)$ 
is small, it might be 
too demanding to choose a very small $q$ since 
$X_{n+j}$ for which $\PP(Y_{n+j}>0\given X_{n+j})$ 
is very large may not exist. 
A practical choice might be some $q$ 
that is moderately larger than the marginal 
proportion of positive 
$Y$'s in the training data (our method 
is still valid). 
In this way, 
when there exists some region in $\cX$ where $\PP(Y >0\given X )\geq 1-q$, 
our method finds critical subsets while 
achieving some power.

\section{Real data application}

\subsection{Candidate screening in recruitment}

We apply $\cfbh$ as an automatic screening tool 
in recruiting, where 
a human resources staff 
uses machine learning prediction to 
screen all applicants and shortlist some for subsequent 
test and interviews. 
In this application, machine learning is used to 
predict whether a new candidate 
is qualified for the job (i.e., whether the recruitment is successful); 
a higher predicted value might indicate a better fit to the position, 
but no guarantee can be provided for a black-box prediction machine. 
We will use $\cfbh$ to 
calibrate the prediction and 
generate a shortlist of candidates with rigorous FDR control, 
i.e., limiting  the proportion of unqualified individuals among the 
selected candidates.  

We assume new applicants to the position 
and previous applicants on record (such as those who applied
last year) are i.i.d.~from the same distribution. 
This is reasonable if the pool of applicants for the position 
is stable over the years. 
The recruiters may train any prediction model 
on  previous applicants 
and use any monotone nonconformity score as their choice.  
We use a small-scale recruitment dataset from Kaggle~\citep{datarecruit}, 
as recruitment datasets from companies are often confidential. 
There are $n_{\textrm{tot}} = 215$ samples in total. 
Each sample is from an applicant 
for the position; the data includes 
covariates about their education, work experience, gender, specialization, etc., 
and the response is a binary variable indicating whether 
the applicant is finally offered the job. 
Here, we use this binary outcome as a perfect proxy for the qualification of a candidate. 
We randomly split the data into a training set of size $|\cD_\train|=86$ 
and a test set of size $|\cD_\test|=43$. 
We first train a gradient boosting model 
to predict the job offering, 
using the \texttt{sckit-learn} Python 
library without fine tuning, 
and apply the three procedures (except Bonferroni since it is less powerful) 
in Section~\ref{sec:simu} for 
$q=\in\{0.1,0.2,0.5\}$.   
We plot the false discovery proportion (FDP) 
and power 
over $N=100$ independent runs in Figures~\ref{fig:recruit} and~\ref{fig:recruit_power}, 
respectively. 

\begin{figure}[h]
  \centering 
  \includegraphics[width=5.5in]{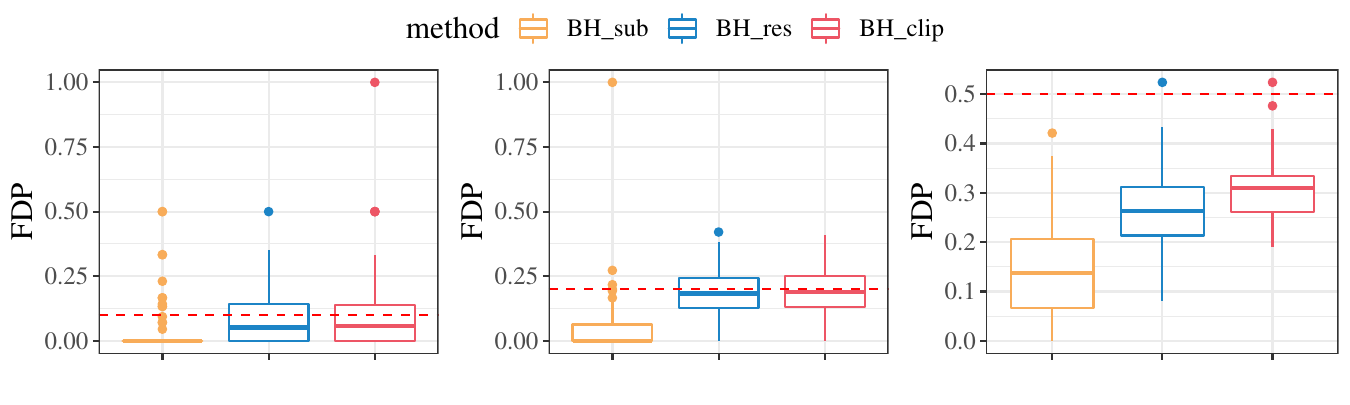}
  \caption{Boxplot for false discovery proportions over $N=100$ independent runs 
  for the recruitment dataset, with $q=0.1$ (left), $q=0.2$ (middle), 
  and $q=0.5$ (right). Red dashed lines are the nominal levels.}
  \label{fig:recruit}
\end{figure}

\begin{figure}[htbp]
  \centering 
  \includegraphics[width=5.5in]{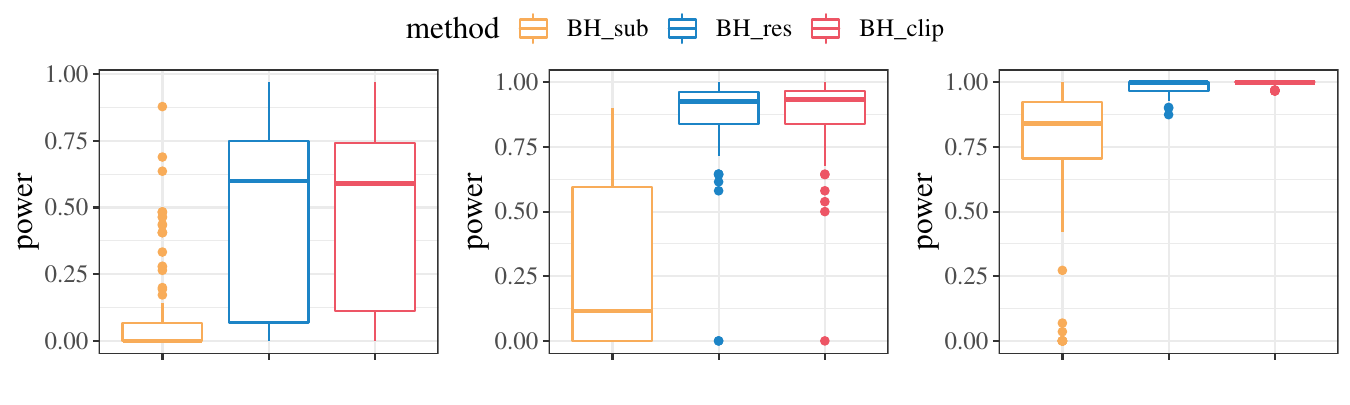}
  \caption{Boxplot for power over $N=100$ independent runs 
  for the recruitment dataset, with $q=0.1$ (left), $q=0.2$ (middle), 
  and $q=0.5$ (right).}
  \label{fig:recruit_power}
\end{figure}

All three scores achieve valid FDR control (averaging the FDPs). 
\texttt{BH\_res} and \texttt{BH\_clip} have similar FDP (hence FDR). 
The FDP of
\texttt{BH\_sub} is lower, potentially because of its low power, 
and also less stable than the other two. 
As \texttt{BH\_sub} only uses the subset of training samples with $Y=0$ 
to calibrate the selection set as~\cite{bates2021testing,rava2021burden} 
did (see discussion in Appendix~\ref{app:subsec_variant}), when such subset is small, 
the calibration, hence the FDP, can be unstable. 
Intuitively, 
\texttt{BH\_res} and \texttt{BH\_clip} achieve a more stable behavior
by using all the calibration data.

The power of these methods differ more significantly. 
In general, predicting qualified candidates from our data 
is a relatively easy task: once 
we allow the FDR level to be $0.2$ or $0.5$, 
\texttt{BH\_res} and \texttt{BH\_clip} could almost 
identify all qualified candidates. 
Both these methods achieve similar power 
while \texttt{BH\_clip} is again a little better. 
However, \texttt{BH\_sub}, which only uses training samples with $Y=0$ 
to calibrate the selection set as~\cite{bates2021testing}
and~\cite{rava2021burden} 
did, has much lower power. 
We discuss this issue in Appendix~\ref{app:sub_variant}: 
in finding positives for a binary response, 
our method can be more powerful than \texttt{BH\_sub} 
when there are many positive samples 
in the population.

\subsection{Drug discovery}

We apply $\cfbh$ to therapeutic datasets for drug discovery, 
focusing on two tasks: (i) selecting molecules that bind  to 
a target protein for a certain disease, 
and (ii) selecting drug-target  (molecule-protein) 
pairs with a high affinity score. 
Our main focus is to calibrate any given prediction model to 
limit false positives. Therefore, we use the pre-trained models and the 
prediction pipelines established in the DeepPurpose library~\citep{huang2020deeppurpose}.

\subsubsection{Drug property prediction for HIV}
\label{subsubsec:hiv}

We first consider the task of predicting drug properties
for a certain protein target for HIV. 
As we mentioned in the introduction, 
given a specific target, 
machine learning models 
are often trained on a representative subset of the whole drug library screened by 
HTS, and then used to predict the activity 
of the remaining proteins to find promising candidates.  
It is important to control for false positives in the shortlisted candidates.

We use the HIV screening dataset with a total size of $n_{\textrm{tot}}=41127$.
We randomly split the data into three folds with ratio $6:2:2$ in size. 
The first two folds contain binary outcomes 
indicating whether the drugs interact  with the disease. 
We use the first fold to train a machine learning model 
to predict the outcome, 
where the 
drugs are encoded into numerical features 
using Extended-Connectivity 
FingerPrints (ECFP) that characterize topological properties of molecules and compounds. 
We train a small neural network in only 3 epochs 
so that the whole procedure works well with CPUs; 
using more complicated or pre-trained 
networks might improve power but this is not the main focus here. 
The second fold serves as the calibration data. 
Our goal is to find active proteins in the last test fold 
while controlling for the false discovery rate. 

In the training fold, about $3\%$ of the drugs 
are active for the HIV disease. 
We choose the FDR levels among $q\in\{0.1,0.2,0.5\}$. 
We compute the empirical FDR, power, and average size of the selection set 
over $N=100$ independent runs 
of the procedure in Table~\ref{tab:protein}. 

\begin{table}[htbp]
\centering
{
\renewcommand\arraystretch{1.2}
\begin{tabular}{c|c |c |c|c|c|c|c|c|c}
\hline
  & \multicolumn{3}{c|}{FDR} & \multicolumn{3}{c|}{Power}& \multicolumn{3}{c}{$|\cR|$}\\
  \hline 
  Level $q$ & 0.1 & 0.2  & 0.5 & 0.1 & 0.2  & 0.5 & 0.1 & 0.2  & 0.5 \\
\hline
\texttt{BH\_clip} & 0.0957 & 0.196 & 0.495 & 0.0788 & 0.174  & 0.410 & 26.5 & 64.2 & 240  \\
\texttt{BH\_res} & 0.0989  & 0.196 &  0.494  & 0.0766 & 0.174  & 0.410 & 25.8 & 64.4 &239 \\
\texttt{BH\_sub} &0.0862 & 0.192 & 0.474 & 0.0739 & 0.169 &  0.397 & 24.8 & 61.8 & 222  \\
\hline
\end{tabular}
}%
\vspace{0.1em}
\caption{FDR and power of the three methods averaged over $N=100$ random splits.}
\vspace{-1em}
\label{tab:protein}
\end{table}

All three choices of nonconformity scores control FDR below the 
nominal levels. 
Their performance is also similar, while 
\texttt{BH\_clip} has the highest power 
and  \texttt{BH\_sub} is the least powerful, both with a small margin. 
This is because the positive samples in the population is extremely small, 
so that using $Y=0$ samples or the whole calibration set 
does not have a huge impact on the selection set. 

Using all three methods, the selection set consists of 
all test samples whose predicted binding affinity is above some value. 
This value is specific to 
the training model we use. 
Figure~\ref{fig:protein} shows the 
selection threshold of the predicted value for all configurations. 
If we control FDR at $q=0.1$, the predicted scores needs to be 
as large as $0.8$ to be considered  promising; 
this leads to around $25$ candidates among about $8000$ test samples. 
However, if we set $q=0.5$, then the thresholds are 
in the range $[0.2,0.4]$ 
most of the time: 
a moderately large score is sufficient to stand out. 

\begin{figure}[htbp]
  \centering 
  \includegraphics[width=5in]{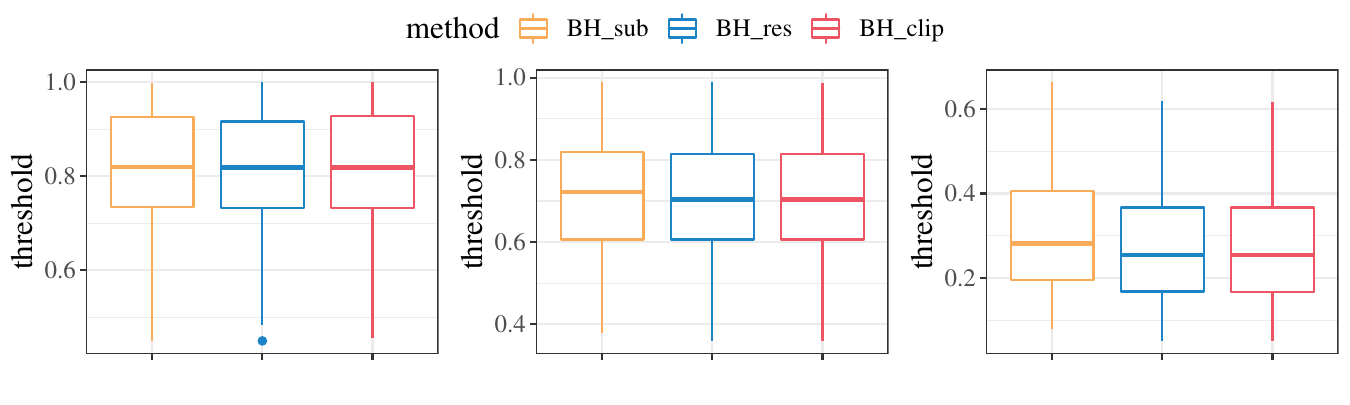}
  \caption{Selection thresholds over $N=100$ runs, 
  with $q=0.1$ (left), $q=0.2$ (middle), 
  and $q=0.5$ (right).}
  \label{fig:protein}
\end{figure}

\subsubsection{Drug-target interaction (DTI) prediction}

Last, we consider the task of predicting drug-target interactions (DTI) 
among a huge pool of drug-target pairs. 
This might be of use to a therapeutic company to 
prioritize its resources in developing drugs that might be 
effective for any of the targets they happen to be interested in. 
In this application, one may need to be 
cautious about the i.i.d.~assumption: this can be reasonable if 
the drugs and targets are drawn from a diverse library, 
and a representative subset of all pairs have been screened 
to form the training data. 

We use the DAVIS dataset published in~\cite{davis2011comprehensive}, 
which records real-valued  binding affinities 
for $n_{\textrm{tot}}=30060$ drug-target pairs.  
In this application, we mimic a scenario where a small proportion of 
the whole library has been screened, 
and one would like to find promising ones among a huge amount of pairs 
whose binding affinities are unknown. 
In particular, we randomly split the dataset into three folds of size
$2:2:6$; we use the first fold for training the model, 
the scond for calibration, and the last largest set as test samples. 
We use ECFP and 
Conjoint triad feature (CTF)~\citep{shen2007predicting,shao2009predicting} 
to encode the drugs and the targets into numeric features, respectively.
We train a small neural network over $10$ epochs. 
These choices are suitable for experiments on CPUs (one might of course 
use other more computationally intensive alternatives). 

Because the affinity is continuously valued, 
and to account for the heterogeneity in targets, 
we set $c_j$ as the $q_{\text{pop}}$-th quantile 
of the outcomes of the training samples with the same 
binding target as sample $j$, where $q_{\text{pop}} \in \{0.7, 0.8, 0.9\}$. 
Given a predicted score, there is no natural way to use \texttt{BH\_sub} in this setting, 
so we test the following two methods: 
\begin{itemize}[itemsep=-0.5ex]
  \item \texttt{BH\_res} with nonconformity score 
  $V(x,y) = y-\hat{\mu}(x)$, 
  \item \texttt{BH\_clip} with nonconformity score 
  $V(x,y; c) = M\ind\{y\geq c \} + c\ind\{y<c\}$ for $M=100$,  
\end{itemize}
where $x$ is the vector of features, $y$ is the binding affinity ranging in $[5,10]$, 
and $c$ is the threshold that is computable for both calibration and test samples. 
We set the FDR level at $q\in\{0.1,0.2,0.5\}$. 
There are $18$ configurations in total, with 2 nonconformity scores 
and $3\times 3$ combinations of $(q_{\text{pop}},q)$. 
In this case, since the threshold $c_j$ varies among the samples, 
the selection is not monotone in the predicted score. 

The empirical FDP, computed as $\frac{\sum_{j=1}^m \ind\{Y_{n+j}\leq c_j\}}{1\vee|\cR|}$ 
($\cR$ is the selection set), 
over $N=100$ independent runs 
is plotted in Figure~\ref{fig:dti_fdr}. 
For all configurations of $q_{\textrm{pop}}$, 
both methods control the FDR (average of FDPs) at the nominal level. 
However, there can be some variation in the FDP 
for $q=0.1$; \texttt{BH\_res} is less stable than \texttt{BH\_clip}. 
Also, the FDR from \texttt{BH\_clip} is very close to the nominal level 
while that from \texttt{BH\_res} is much lower. 
This is due to the low power of \texttt{BH\_res} as we show 
in the power plot (Figure~\ref{fig:dti_power}). 

\begin{figure}[h]
  \centering 
  \includegraphics[width=5.5in]{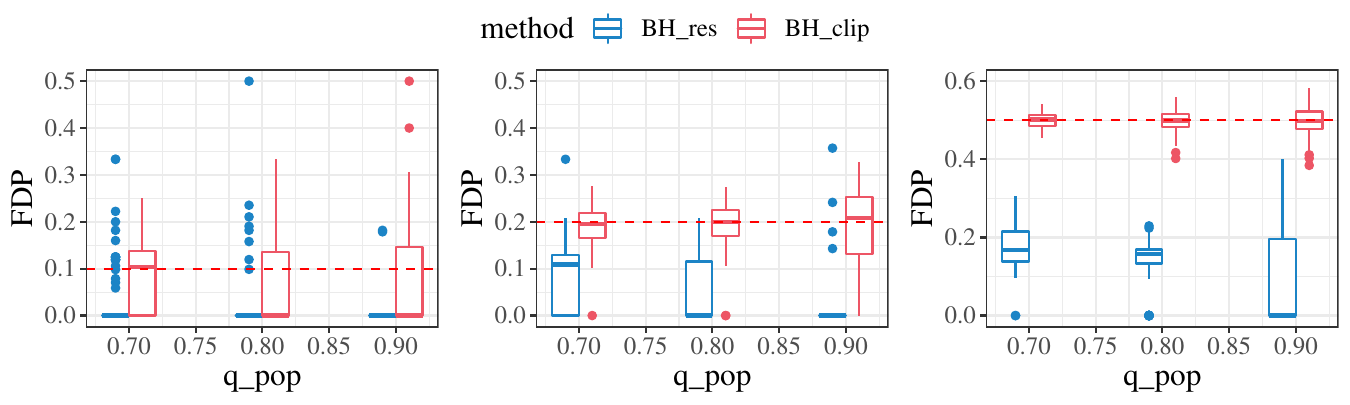}
  \caption{FDP over $N=100$ runs, 
  with $q=0.1$ (left), $q=0.2$ (middle),  $q=0.5$ (right).}
  \label{fig:dti_fdr}
\end{figure}

The power of both methods decrease as $q_{\textrm{pop}}$ increases, 
which is natural because this leads to higher thresholds for binding affinity. 
In all settings, \texttt{BH\_clip} is more powerful and also 
yields larger selection sets. 
Thus, when the thresholds are computable for both the 
calibration and test samples (as Example~\ref{ex:random}), 
we recommend \texttt{BH\_clip} for higher power. 

\begin{figure}[H]
  \centering 
  \includegraphics[width=5.5in]{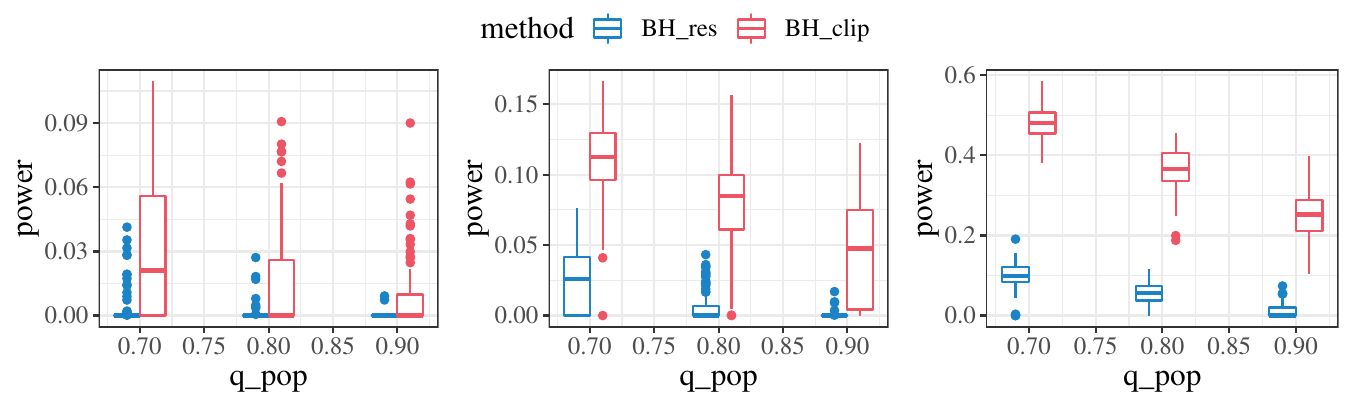}
  \caption{Power over $N=100$ runs, 
  with $q=0.1$ (left), $q=0.2$ (middle), 
  $q=0.5$ (right).}
  \label{fig:dti_power}
\end{figure}

\section{Discussion}

In this paper, we introduce $\cfbh$, which is a generic tool 
to turn any prediction model into a selection threshold for 
interesting outcomes. 
By constructing conformal p-values based on i.i.d.~calibration data 
and leveraging multiple testing ideas, 
we guarantee that a prescribed proportion of the selected set 
is indeed of interest. 
Controlling the false discovery rate ensures 
efficient use of resources for follow-up investigations. 

A crucial condition that $\cfbh$ relies on 
is that the calibration and test samples are i.i.d.~or exchangeable. 
However, in practice, the two datasets might differ 
because of selection or distribution shift. 
For example, to infer the performance of this year's job candidates, 
last years' candidates that are documented 
might  in general be  more competent than 
average; 
to infer new drugs, the 
drugs that have been screened by HTS might be selected with 
varying preference based on the features; 
drug discovery also needs to deal with domain shift (repurposing) 
for completely unseen targets. 
Reliable selection under distribution shift, if not infeasible,  
may require more involved techniques. 

FDR, as a measure of Type-I error, may be limited in applications such as 
healthcare, where both type-I and type-II errors are of concern. 
Therefore, it might also be interesting to see whether our methodology 
can be extended to controlling a mixture of both error types. 
Meanwhile, counting the number of errors may be less sensible if 
the cost of making an error varies with individuals or depends on the outcomes. 
Developing calibration methods to control general risks in screening procedures 
is also an interesting direction to pursue. 

\section*{Acknowledgement}
We are grateful to two anonymous referees and the action editor for valuable comments and suggestions.
We thank John Cherian, Issac Gibbs, 
Jayoon Jang, Lihua Lei, Shuangning Li, Zhimei Ren, 
Hui Xu, and Qian Zhao for helpful discussions. 
Y.J.~would like to specially thank John Cherian for inspiring discussion  
on the applications. 
E.C.~and Y.J.~were supported by the Office of Naval Research grant N00014-20-1-2157, the National Science
Foundation grant DMS-2032014, the Simons Foundation under award 814641, and the ARO grant
2003514594. 

\newpage
\bibliographystyle{plainnat}
\bibliography{reference}

\begin{thebibliography}{63}
\providecommand{\natexlab}[1]{#1}
\providecommand{\url}[1]{\texttt{#1}}
\expandafter\ifx\csname urlstyle\endcsname\relax
  \providecommand{\doi}[1]{doi: #1}\else
  \providecommand{\doi}{doi: \begingroup \urlstyle{rm}\Url}\fi

\bibitem[Ahlberg et~al.(2017{\natexlab{a}})Ahlberg, Hammar, Bendtsen, and
  Carlsson]{ahlberg2017current}
Ernst Ahlberg, Oscar Hammar, Claus Bendtsen, and Lars Carlsson.
\newblock Current application of conformal prediction in drug discovery.
\newblock \emph{Annals of Mathematics and Artificial Intelligence}, 81\penalty0
  (1):\penalty0 145--154, 2017{\natexlab{a}}.

\bibitem[Ahlberg et~al.(2017{\natexlab{b}})Ahlberg, Winiwarter, Bostr{\"o}m,
  Linusson, L{\"o}fstr{\"o}m, Johansson, Engkvist, Hammar, Bendtsen, and
  Carlsson]{ahlberg2017using}
Ernst Ahlberg, Susanne Winiwarter, Henrik Bostr{\"o}m, Henrik Linusson, Tuve
  L{\"o}fstr{\"o}m, Ulf Norinder~Ulf Johansson, Ola Engkvist, Oscar Hammar,
  Claus Bendtsen, and Lars Carlsson.
\newblock Using conformal prediction to prioritize compound synthesis in drug
  discovery.
\newblock In \emph{Conformal and Probabilistic Prediction and Applications},
  pages 174--184. PMLR, 2017{\natexlab{b}}.

\bibitem[Amdouni and abdessalem Karaa(2010)]{amdouni2010web}
Soumaya Amdouni and Wahiba~Ben abdessalem Karaa.
\newblock Web-based recruiting.
\newblock In \emph{ACS/IEEE International Conference on Computer Systems and
  Applications-AICCSA 2010}, pages 1--7. IEEE, 2010.

\bibitem[Bates et~al.(2021)Bates, Cand{\`e}s, Lei, Romano, and
  Sesia]{bates2021testing}
Stephen Bates, Emmanuel Cand{\`e}s, Lihua Lei, Yaniv Romano, and Matteo Sesia.
\newblock Testing for outliers with conformal p-values.
\newblock \emph{arXiv preprint arXiv:2104.08279}, 2021.

\bibitem[Benjamini and Hochberg(1995)]{benjamini1995controlling}
Yoav Benjamini and Yosef Hochberg.
\newblock Controlling the false discovery rate: a practical and powerful
  approach to multiple testing.
\newblock \emph{Journal of the Royal statistical society: series B
  (Methodological)}, 57\penalty0 (1):\penalty0 289--300, 1995.

\bibitem[Benjamini and Hochberg(1997)]{benjamini1997multiple}
Yoav Benjamini and Yosef Hochberg.
\newblock Multiple hypotheses testing with weights.
\newblock \emph{Scandinavian Journal of Statistics}, 24\penalty0 (3):\penalty0
  407--418, 1997.

\bibitem[Benjamini and Yekutieli(2001)]{benjamini2001control}
Yoav Benjamini and Daniel Yekutieli.
\newblock The control of the false discovery rate in multiple testing under
  dependency.
\newblock \emph{Annals of statistics}, pages 1165--1188, 2001.

\bibitem[Bohrer(1979)]{bohrer1979multiple}
Robert Bohrer.
\newblock Multiple three-decision rules for parametric signs.
\newblock \emph{Journal of the American Statistical Association}, 74\penalty0
  (366a):\penalty0 432--437, 1979.

\bibitem[Bohrer and Schervish(1980)]{bohrer1980optimal}
Robert Bohrer and Mark~J Schervish.
\newblock An optimal multiple decision rule for signs of parameters.
\newblock \emph{Proceedings of the National Academy of Sciences}, 77\penalty0
  (1):\penalty0 52--56, 1980.

\bibitem[Cand{\`e}s et~al.(2021)Cand{\`e}s, Lei, and
  Ren]{candes2021conformalized}
Emmanuel~J Cand{\`e}s, Lihua Lei, and Zhimei Ren.
\newblock Conformalized survival analysis.
\newblock \emph{arXiv preprint arXiv:2103.09763}, 2021.

\bibitem[Carracedo-Reboredo et~al.(2021)Carracedo-Reboredo, Li{\~n}ares-Blanco,
  Rodr{\'\i}guez-Fern{\'a}ndez, Cedr{\'o}n, Novoa, Carballal, Maojo, Pazos, and
  Fernandez-Lozano]{carracedo2021review}
Paula Carracedo-Reboredo, Jose Li{\~n}ares-Blanco, Nereida
  Rodr{\'\i}guez-Fern{\'a}ndez, Francisco Cedr{\'o}n, Francisco~J Novoa, Adrian
  Carballal, Victor Maojo, Alejandro Pazos, and Carlos Fernandez-Lozano.
\newblock A review on machine learning approaches and trends in drug discovery.
\newblock \emph{Computational and Structural Biotechnology Journal},
  19:\penalty0 4538--4558, 2021.

\bibitem[Chernozhukov et~al.(2021)Chernozhukov, W{\"u}thrich, and
  Zhu]{chernozhukov2021distributional}
Victor Chernozhukov, Kaspar W{\"u}thrich, and Yinchu Zhu.
\newblock Distributional conformal prediction.
\newblock \emph{Proceedings of the National Academy of Sciences}, 118\penalty0
  (48), 2021.

\bibitem[Corey et~al.(2018)Corey, Kashyap, Lorenzi, Lagoo-Deenadayalan, Heller,
  Whalen, Balu, Heflin, McDonald, Swaminathan, et~al.]{corey2018development}
Kristin~M Corey, Sehj Kashyap, Elizabeth Lorenzi, Sandhya~A Lagoo-Deenadayalan,
  Katherine Heller, Krista Whalen, Suresh Balu, Mitchell~T Heflin, Shelley~R
  McDonald, Madhav Swaminathan, et~al.
\newblock Development and validation of machine learning models to identify
  high-risk surgical patients using automatically curated electronic health
  record data (pythia): A retrospective, single-site study.
\newblock \emph{PLoS medicine}, 15\penalty0 (11):\penalty0 e1002701, 2018.

\bibitem[Cort{\'e}s-Ciriano and Bender(2019)]{cortes2019concepts}
Isidro Cort{\'e}s-Ciriano and Andreas Bender.
\newblock Concepts and applications of conformal prediction in computational
  drug discovery.
\newblock \emph{arXiv preprint arXiv:1908.03569}, 2019.

\bibitem[Dara et~al.(2021)Dara, Dhamercherla, Jadav, Babu, and
  Ahsan]{dara2021machine}
Suresh Dara, Swetha Dhamercherla, Surender~Singh Jadav, CH~Babu, and
  Mohamed~Jawed Ahsan.
\newblock Machine learning in drug discovery: A review.
\newblock \emph{Artificial Intelligence Review}, pages 1--53, 2021.

\bibitem[Davis et~al.(2011)Davis, Hunt, Herrgard, Ciceri, Wodicka, Pallares,
  Hocker, Treiber, and Zarrinkar]{davis2011comprehensive}
Mindy~I Davis, Jeremy~P Hunt, Sanna Herrgard, Pietro Ciceri, Lisa~M Wodicka,
  Gabriel Pallares, Michael Hocker, Daniel~K Treiber, and Patrick~P Zarrinkar.
\newblock Comprehensive analysis of kinase inhibitor selectivity.
\newblock \emph{Nature biotechnology}, 29\penalty0 (11):\penalty0 1046--1051,
  2011.

\bibitem[Dickhaus(2014)]{dickhaus2014multiple}
Thorsten Dickhaus.
\newblock Multiple testing and binary classification.
\newblock In \emph{Simultaneous Statistical Inference}, pages 91--101.
  Springer, 2014.

\bibitem[Eklund et~al.(2015)Eklund, Norinder, Boyer, and
  Carlsson]{eklund2015application}
Martin Eklund, Ulf Norinder, Scott Boyer, and Lars Carlsson.
\newblock The application of conformal prediction to the drug discovery
  process.
\newblock \emph{Annals of Mathematics and Artificial Intelligence}, 74\penalty0
  (1):\penalty0 117--132, 2015.

\bibitem[Etzioni et~al.(2003)Etzioni, Urban, Ramsey, McIntosh, Schwartz, Reid,
  Radich, Anderson, and Hartwell]{etzioni2003case}
Ruth Etzioni, Nicole Urban, Scott Ramsey, Martin McIntosh, Stephen Schwartz,
  Brian Reid, Jerald Radich, Garnet Anderson, and Leland Hartwell.
\newblock The case for early detection.
\newblock \emph{Nature reviews cancer}, 3\penalty0 (4):\penalty0 243--252,
  2003.

\bibitem[Faliagka et~al.(2012)Faliagka, Ramantas, Tsakalidis, and
  Tzimas]{faliagka2012application}
Evanthia Faliagka, Kostas Ramantas, Athanasios Tsakalidis, and Giannis Tzimas.
\newblock Application of machine learning algorithms to an online recruitment
  system.
\newblock In \emph{Proc. International Conference on Internet and Web
  Applications and Services}, pages 215--220. Citeseer, 2012.

\bibitem[Faliagka et~al.(2014)Faliagka, Iliadis, Karydis, Rigou, Sioutas,
  Tsakalidis, and Tzimas]{faliagka2014line}
Evanthia Faliagka, Lazaros Iliadis, Ioannis Karydis, Maria Rigou, Spyros
  Sioutas, Athanasios Tsakalidis, and Giannis Tzimas.
\newblock On-line consistent ranking on e-recruitment: seeking the truth behind
  a well-formed cv.
\newblock \emph{Artificial Intelligence Review}, 42\penalty0 (3):\penalty0
  515--528, 2014.

\bibitem[FDA(2018)]{fdadrug}
U.S. FDA.
\newblock The drug development process, 2018.
\newblock URL
  \url{https://www.fda.gov/patients/learn-about-drug-and-device-approvals/drug-development-process}.

\bibitem[Guo et~al.(2010)Guo, Sarkar, and Peddada]{guo2010controlling}
Wenge Guo, Sanat~K Sarkar, and Shyamal~D Peddada.
\newblock Controlling false discoveries in multidimensional directional
  decisions, with applications to gene expression data on ordered categories.
\newblock \emph{Biometrics}, 66\penalty0 (2):\penalty0 485--492, 2010.

\bibitem[Hastie et~al.(2009)Hastie, Tibshirani, Friedman, and
  Friedman]{hastie2009elements}
Trevor Hastie, Robert Tibshirani, Jerome~H Friedman, and Jerome~H Friedman.
\newblock \emph{The elements of statistical learning: data mining, inference,
  and prediction}, volume~2.
\newblock Springer, 2009.

\bibitem[Heaslip(2022)]{hirearticle}
Emily Heaslip.
\newblock Ai tools for talent acquisition to help you hire, 2022.
\newblock URL \url{https://vervoe.com/ai-tools-for-talent-acquisition/}.

\bibitem[Hochberg(1986)]{hochberg1986multiple}
Yosef Hochberg.
\newblock Multiple classification rules for signs of parameters.
\newblock \emph{Journal of statistical planning and inference}, 15:\penalty0
  177--188, 1986.

\bibitem[Huang et~al.(2020)Huang, Fu, Glass, Zitnik, Xiao, and
  Sun]{huang2020deeppurpose}
Kexin Huang, Tianfan Fu, Lucas~M Glass, Marinka Zitnik, Cao Xiao, and Jimeng
  Sun.
\newblock Deeppurpose: A deep learning library for drug-target interaction
  prediction.
\newblock \emph{Bioinformatics}, 2020.

\bibitem[Huang(2007)]{huang2007drug}
Ziwei Huang.
\newblock \emph{Drug discovery research: new frontiers in the post-genomic
  era}.
\newblock John Wiley \& Sons, 2007.

\bibitem[Imbens and Rubin(2015)]{imbens2015causal}
Guido~W. Imbens and Donald~B. Rubin.
\newblock \emph{Causal Inference for Statistics, Social, and Biomedical
  Sciences: An Introduction}.
\newblock Cambridge University Press, 2015.
\newblock \doi{10.1017/CBO9781139025751}.

\bibitem[Jalali et~al.(2020)Jalali, Lonsdale, Do, Peck, Gupta, Kutty,
  Ghazarian, Jacobs, Rehman, and Ahumada]{jalali2020deep}
Ali Jalali, Hannah Lonsdale, Nhue Do, Jacquelin Peck, Monesha Gupta, Shelby
  Kutty, Sharon~R Ghazarian, Jeffrey~P Jacobs, Mohamed Rehman, and Luis~M
  Ahumada.
\newblock Deep learning for improved risk prediction in surgical outcomes.
\newblock \emph{Scientific reports}, 10\penalty0 (1):\penalty0 1--13, 2020.

\bibitem[Jin et~al.(2023)Jin, Ren, and Cand{\`e}s]{jin2023sensitivity}
Ying Jin, Zhimei Ren, and Emmanuel~J Cand{\`e}s.
\newblock Sensitivity analysis of individual treatment effects: A robust
  conformal inference approach.
\newblock \emph{Proceedings of the National Academy of Sciences}, 120\penalty0
  (6):\penalty0 e2214889120, 2023.

\bibitem[Kim et~al.(2021)Kim, Chen, Cheng, Gindulyte, He, He, Li, Shoemaker,
  Thiessen, Yu, et~al.]{kim2021pubchem}
Sunghwan Kim, Jie Chen, Tiejun Cheng, Asta Gindulyte, Jia He, Siqian He,
  Qingliang Li, Benjamin~A Shoemaker, Paul~A Thiessen, Bo~Yu, et~al.
\newblock Pubchem in 2021: new data content and improved web interfaces.
\newblock \emph{Nucleic acids research}, 49\penalty0 (D1):\penalty0
  D1388--D1395, 2021.

\bibitem[Koutsoukas et~al.(2017)Koutsoukas, Monaghan, Li, and
  Huan]{koutsoukas2017deep}
Alexios Koutsoukas, Keith~J Monaghan, Xiaoli Li, and Jun Huan.
\newblock Deep-learning: investigating deep neural networks hyper-parameters
  and comparison of performance to shallow methods for modeling bioactivity
  data.
\newblock \emph{Journal of cheminformatics}, 9\penalty0 (1):\penalty0 1--13,
  2017.

\bibitem[Lampa et~al.(2018)Lampa, Alvarsson, Arvidsson Mc~Shane, Berg, Ahlberg,
  and Spjuth]{lampa2018predicting}
Samuel Lampa, Jonathan Alvarsson, Staffan Arvidsson Mc~Shane, Arvid Berg, Ernst
  Ahlberg, and Ola Spjuth.
\newblock Predicting off-target binding profiles with confidence using
  conformal prediction.
\newblock \emph{Frontiers in pharmacology}, 9:\penalty0 1256, 2018.

\bibitem[Lei et~al.(2018)Lei, G’Sell, Rinaldo, Tibshirani, and
  Wasserman]{lei2018distribution}
Jing Lei, Max G’Sell, Alessandro Rinaldo, Ryan~J Tibshirani, and Larry
  Wasserman.
\newblock Distribution-free predictive inference for regression.
\newblock \emph{Journal of the American Statistical Association}, 113\penalty0
  (523):\penalty0 1094--1111, 2018.

\bibitem[Lei and Cand{\`e}s(2021)]{lei2021conformal}
Lihua Lei and Emmanuel~J Cand{\`e}s.
\newblock Conformal inference of counterfactuals and individual treatment
  effects.
\newblock \emph{Journal of the Royal Statistical Society: Series B (Statistical
  Methodology)}, 2021.

\bibitem[Lindh et~al.(2017)Lindh, Karl{\'e}n, and
  Norinder]{lindh2017predicting}
Martin Lindh, Anders Karl{\'e}n, and Ulf Norinder.
\newblock Predicting the rate of skin penetration using an aggregated conformal
  prediction framework.
\newblock \emph{Molecular Pharmaceutics}, 14\penalty0 (5):\penalty0 1571--1576,
  2017.

\bibitem[Macarron et~al.(2011)Macarron, Banks, Bojanic, Burns, Cirovic,
  Garyantes, Green, Hertzberg, Janzen, Paslay, et~al.]{macarron2011impact}
Ricardo Macarron, Martyn~N Banks, Dejan Bojanic, David~J Burns, Dragan~A
  Cirovic, Tina Garyantes, Darren~VS Green, Robert~P Hertzberg, William~P
  Janzen, Jeff~W Paslay, et~al.
\newblock Impact of high-throughput screening in biomedical research.
\newblock \emph{Nature reviews Drug discovery}, 10\penalty0 (3):\penalty0
  188--195, 2011.

\bibitem[Mary and Roquain(2021)]{mary2021semi}
David Mary and Etienne Roquain.
\newblock Semi-supervised multiple testing.
\newblock \emph{arXiv preprint arXiv:2106.13501}, 2021.

\bibitem[Mochol et~al.(2007)Mochol, Wache, and Nixon]{mochol2007improving}
Malgorzata Mochol, Holger Wache, and Lyndon Nixon.
\newblock Improving the accuracy of job search with semantic techniques.
\newblock In \emph{International Conference on Business Information Systems},
  pages 301--313. Springer, 2007.

\bibitem[Rahimi et~al.(2014)Rahimi, Bennett, Conrad, Williams, Basu, Dwight,
  Woodward, Patel, McMurray, and MacMahon]{rahimi2014risk}
Kazem Rahimi, Derrick Bennett, Nathalie Conrad, Timothy~M Williams, Joyee Basu,
  Jeremy Dwight, Mark Woodward, Anushka Patel, John McMurray, and Stephen
  MacMahon.
\newblock Risk prediction in patients with heart failure: a systematic review
  and analysis.
\newblock \emph{JACC: Heart Failure}, 2\penalty0 (5):\penalty0 440--446, 2014.

\bibitem[Rava et~al.(2021)Rava, Sun, James, and Tong]{rava2021burden}
Bradley Rava, Wenguang Sun, Gareth~M James, and Xin Tong.
\newblock A burden shared is a burden halved: A fairness-adjusted approach to
  classification.
\newblock \emph{arXiv preprint arXiv:2110.05720}, 2021.

\bibitem[Richens et~al.(2020)Richens, Lee, and Johri]{richens2020improving}
Jonathan~G Richens, Ciar{\'a}n~M Lee, and Saurabh Johri.
\newblock Improving the accuracy of medical diagnosis with causal machine
  learning.
\newblock \emph{Nature communications}, 11\penalty0 (1):\penalty0 1--9, 2020.

\bibitem[Romano et~al.(2019)Romano, Patterson, and
  Candes]{romano2019conformalized}
Yaniv Romano, Evan Patterson, and Emmanuel Candes.
\newblock Conformalized quantile regression.
\newblock \emph{Advances in neural information processing systems}, 32, 2019.

\bibitem[Roquain and Verzelen(2022)]{roquain2022false}
Etienne Roquain and Nicolas Verzelen.
\newblock False discovery rate control with unknown null distribution: Is it
  possible to mimic the oracle?
\newblock \emph{The Annals of Statistics}, 50\penalty0 (2):\penalty0
  1095--1123, 2022.

\bibitem[Roshan(2020)]{datarecruit}
Ben Roshan.
\newblock Campus recruitment, 2020.
\newblock Academic and Employability Factors influencing placement,
  \url{https://www.kaggle.com/datasets/benroshan/factors-affecting-campus-placement}.

\bibitem[Sahoo et~al.(2021)Sahoo, Zhao, Chen, and Ermon]{sahoo2021reliable}
Roshni Sahoo, Shengjia Zhao, Alyssa Chen, and Stefano Ermon.
\newblock Reliable decisions with threshold calibration.
\newblock \emph{Advances in Neural Information Processing Systems}, 34, 2021.

\bibitem[Scott et~al.(2009)Scott, Bellala, and Willett]{scott2009false}
Clayton Scott, Gowtham Bellala, and Rebecca Willett.
\newblock The false discovery rate for statistical pattern recognition.
\newblock \emph{Electronic Journal of Statistics}, 3:\penalty0 651--677, 2009.

\bibitem[Shao et~al.(2009)Shao, Tian, Wu, Wang, Jing, and
  Deng]{shao2009predicting}
Xiaojian Shao, Yingjie Tian, Lingyun Wu, Yong Wang, Ling Jing, and Naiyang
  Deng.
\newblock Predicting dna-and rna-binding proteins from sequences with kernel
  methods.
\newblock \emph{Journal of Theoretical Biology}, 258\penalty0 (2):\penalty0
  289--293, 2009.

\bibitem[Shehu and Saeed(2016)]{shehu2016adaptive}
Muhammad~Ahmad Shehu and Faisal Saeed.
\newblock An adaptive personnel selection model for recruitment using
  domain-driven data mining.
\newblock \emph{Journal of Theoretical and Applied Information Technology},
  91\penalty0 (1):\penalty0 117, 2016.

\bibitem[Shen et~al.(2007)Shen, Zhang, Luo, Zhu, Yu, Chen, Li, and
  Jiang]{shen2007predicting}
Juwen Shen, Jian Zhang, Xiaomin Luo, Weiliang Zhu, Kunqian Yu, Kaixian Chen,
  Yixue Li, and Hualiang Jiang.
\newblock Predicting protein--protein interactions based only on sequences
  information.
\newblock \emph{Proceedings of the National Academy of Sciences}, 104\penalty0
  (11):\penalty0 4337--4341, 2007.

\bibitem[Shen et~al.(2019)Shen, Margolies, Rothstein, Fluder, McBride, and
  Sieh]{shen2019deep}
Li~Shen, Laurie~R Margolies, Joseph~H Rothstein, Eugene Fluder, Russell
  McBride, and Weiva Sieh.
\newblock Deep learning to improve breast cancer detection on screening
  mammography.
\newblock \emph{Scientific reports}, 9\penalty0 (1):\penalty0 1--12, 2019.

\bibitem[Sink et~al.(2010)Sink, Gobec, Pecar, and Zega]{sink2010false}
Roman Sink, Stanislav Gobec, S~Pecar, and Anamarija Zega.
\newblock False positives in the early stages of drug discovery.
\newblock \emph{Current medicinal chemistry}, 17\penalty0 (34):\penalty0
  4231--4255, 2010.

\bibitem[Storey et~al.(2004)Storey, Taylor, and Siegmund]{storey2004strong}
John~D Storey, Jonathan~E Taylor, and David Siegmund.
\newblock Strong control, conservative point estimation and simultaneous
  conservative consistency of false discovery rates: a unified approach.
\newblock \emph{Journal of the Royal Statistical Society: Series B (Statistical
  Methodology)}, 66\penalty0 (1):\penalty0 187--205, 2004.

\bibitem[Svensson et~al.(2017)Svensson, Norinder, and
  Bender]{svensson2017improving}
Fredrik Svensson, Ulf Norinder, and Andreas Bender.
\newblock Improving screening efficiency through iterative screening using
  docking and conformal prediction.
\newblock \emph{Journal of chemical information and modeling}, 57\penalty0
  (3):\penalty0 439--444, 2017.

\bibitem[Svensson et~al.(2018)Svensson, Aniceto, Norinder, Cortes-Ciriano,
  Spjuth, Carlsson, and Bender]{svensson2018conformal}
Fredrik Svensson, Natalia Aniceto, Ulf Norinder, Isidro Cortes-Ciriano, Ola
  Spjuth, Lars Carlsson, and Andreas Bender.
\newblock Conformal regression for quantitative structure--activity
  relationship modeling—quantifying prediction uncertainty.
\newblock \emph{Journal of Chemical Information and Modeling}, 58\penalty0
  (5):\penalty0 1132--1140, 2018.

\bibitem[Szyma{\'n}ski et~al.(2011)Szyma{\'n}ski, Markowicz, and
  Mikiciuk-Olasik]{szymanski2011adaptation}
Pawe{\l} Szyma{\'n}ski, Magdalena Markowicz, and El{\.z}bieta Mikiciuk-Olasik.
\newblock Adaptation of high-throughput screening in drug
  discovery—toxicological screening tests.
\newblock \emph{International journal of molecular sciences}, 13\penalty0
  (1):\penalty0 427--452, 2011.

\bibitem[Tibshirani et~al.(2019)Tibshirani, Foygel~Barber, Candes, and
  Ramdas]{tibshirani2019conformal}
Ryan~J Tibshirani, Rina Foygel~Barber, Emmanuel Candes, and Aaditya Ramdas.
\newblock Conformal prediction under covariate shift.
\newblock \emph{Advances in neural information processing systems}, 32, 2019.

\bibitem[Vamathevan et~al.(2019)Vamathevan, Clark, Czodrowski, Dunham, Ferran,
  Lee, Li, Madabhushi, Shah, Spitzer, et~al.]{vamathevan2019applications}
Jessica Vamathevan, Dominic Clark, Paul Czodrowski, Ian Dunham, Edgardo Ferran,
  George Lee, Bin Li, Anant Madabhushi, Parantu Shah, Michaela Spitzer, et~al.
\newblock Applications of machine learning in drug discovery and development.
\newblock \emph{Nature reviews Drug discovery}, 18\penalty0 (6):\penalty0
  463--477, 2019.

\bibitem[Vovk et~al.(2005)Vovk, Gammerman, and Shafer]{vovk2005algorithmic}
Vladimir Vovk, Alexander Gammerman, and Glenn Shafer.
\newblock \emph{Algorithmic learning in a random world}.
\newblock Springer Science \& Business Media, 2005.

\bibitem[Vovk et~al.(1999)Vovk, Gammerman, and Saunders]{vovk1999machine}
Volodya Vovk, Alexander Gammerman, and Craig Saunders.
\newblock Machine-learning applications of algorithmic randomness.
\newblock 1999.

\bibitem[Wang et~al.(2022)Wang, Joachims, and Rodriguez]{wang2022improving}
Lequn Wang, Thorsten Joachims, and Manuel~Gomez Rodriguez.
\newblock Improving screening processes via calibrated subset selection.
\newblock \emph{arXiv preprint arXiv:2202.01147}, 2022.

\bibitem[Weinstein and Ramdas(2020)]{weinstein2020online}
Asaf Weinstein and Aaditya Ramdas.
\newblock Online control of the false coverage rate and false sign rate.
\newblock In \emph{International Conference on Machine Learning}, pages
  10193--10202. PMLR, 2020.

\end{thebibliography}

\newpage 
\appendix

\section{Connections to the literature}
\label{app:condition_variant}

This section provides a detailed comparison
of our method
to related ones in the literature.
In Section~\ref{app:sub_variant}, we present
a variant of our method that applies to
classification problems and controls FDR conditional on the test responses.
We then compare it to the score-based methods of~\cite{mary2021semi} and~\cite{rava2021burden}
for classification problems in Section~\ref{app:subsec_classify}.
In Section~\ref{app:sub_outlier}, we show
how the outlier detection problem in~\cite{bates2021testing} with conformal p-values
can be turned into the classification setting
and discuss its connection to our methods.

\subsection{A variant for hypotheses-conditional FDR control}
\label{app:sub_variant}

A variant of Algorithm~\ref{alg:bh}
provides a slightly stronger guarantee, FDR control for
test data conditional on $\cH_0$.
This variant applies to
binary or one-versus-all classification; it also applies to
our directional selection with constant thresholds, because it can be turned into
a classification problem.
The classification setting is close to that of~\cite{rava2021burden},
while our analysis through conformal p-values offers a complementary perspective. 

We assume access to i.i.d.~training data $\{(X_i,L_i)\}_{i=1}^n$
where $L_i$ is the label of unit $i$ and
$X_i\in \cX$ is the features. We also
assume the covariates of
the testing sample $\{X_{n+j}\}_{j=1}^m$ are observed, but not the label.
Suppose one is interested
in finding a subset $\cR$  of test samples whose labels are in some
user-specified class $\cC$,
while controlling the FDR, defined as
\$
\fdr := \EE\Bigg[ \frac{ \sum_{j=1}^m \ind\{L_j\notin \cC, j\in \cR\}}{1\vee|\cR|}\Bigg],
\$
where the expectation is with respect to the randomness in
all the training and testing data.
One could encode a binary label $Y = \ind\{L\in \cC\}$,
which turns it into our setting with $c_j\equiv 0.5$ in~\eqref{eq:Hj}.

Since the label is observable in the calibration fold,
we choose a subset $\cD_\calib^0 = \{i\in \cD_\calib\colon Y_i =0\}$
and let $n^0 = |\cD_\calib^0|$. 
We construct conformal p-values
via
\#\label{eq:pj0}
p_j^0 = \frac{\sum_{i\in \cD_\calib^0} \ind {\{ \hat V_i <\hat{V}_{n+j} \}}+ ( 1+ \sum_{i\in \cD_\calib^0} \ind {\{ \hat V_i = \hat{V}_{n+j} \}}) \cdot U_j   }{n^0+1},
\#
where $U_j\sim \text{Unif}\,[0,1]$ are i.i.d.~random variables.
The construction of p-values in~\eqref{eq:pj0}
differs from~\eqref{eq:pj} in that
we use a smaller calibration set $\cD_\calib^0$, and
compare the test nonconformity scores $\hat{V}_{n+j}$
to $\hat{V}_i = V(X_{i},0)$,
instead of $V_i=V(X_i,Y_i)$ for $i\in \cD_\calib^0$.
We then run BH with $\{p_j^0\}$.
We name this procedure as $\cfbhnull$, which is summarized in Algorithm~\ref{alg:bh0}.

\begin{algorithm}[H]
  \caption{$\cfbhnull$: Selection by prediction with same-class calibration}\label{alg:bh0}
  \begin{algorithmic}[1]
  \REQUIRE Calibration data $\{(X_i,Y_i)\}_{i\in \cD_\calib^0}$,
  test data covariates $\{X_{n+j}\}_{j=1}^m$,
  FDR target $q\in(0,1)$, monotone nonconformity score $V\colon \cX\times\cY\to \RR$.
  \vspace{0.05in}
  \STATE Compute $\hat{V}_i = V(X_i,c_j)$ for $i\in \cD_\calib$ and $\hat{V}_{n+j}= V(X_{n+j},0)$ for $j=1,\dots,m$
  \STATE Construct conformal p-values $\{p_j^0\}_{j=1}^m$ as in~\eqref{eq:pj0} 
  \STATE (BH procedure) Compute $k^* = \max\big\{k\colon \sum_{j=1}^m \ind\{p_j^0\leq qk/m\} \geq k\big\}$
  \vspace{0.05in}
  \ENSURE Selection set $\cR = \{j\colon p_j^0 \leq qk^*/m\}$.
  \end{algorithmic}
\end{algorithm}

Using a slightly different argument,
the following proposition shows the FDR control
of $\cfbhnull$ conditional on the labels of test samples.
The proof is in Appendix~\ref{app:subsec_variant}.

\begin{prop}\label{prop:variant}
Suppose $V$ is monotone, $\cD_\calib$ and test data are i.i.d., and $\cY=\{0,1\}$.
Let $\cD_\calib^0 = \{i\in \cD_\calib\colon Y_i=0\}$.
Given any $q\in [0,1]$, the output of Algorithm~\ref{alg:bh0} satisfies
\$
\EE\Bigg[  \frac{\sum_{j=1}^m \ind\{j\in \cR, Y_{n+j}= 0\}}{1\vee|\cR|}  \Bigggiven  \{Y_{n+j}\}_{j=1}^m \Bigg] \leq q.
\$
\end{prop}

\subsection{Connection to classification methods}
\label{app:subsec_classify}

A recent paper~\citep{rava2021burden}
considers controlling the mis-classification rate
(called FSR, the false selection rate), a similar notion as our FDR,
for certain pre-specified subgroups
in classification problems.
Given any model $\hat{s}(x)$ that predicts
how likely the label of a sample
with feature value $x$ is in a specified class $\cC$, they use
an independent set of
data from class $\cC$ to calibrate  a threshold $\hat{s}\in \RR$,
and classify all the test samples
with $\hat{s}(X_{n+j})\geq \hat{s}$ into $\cC$.
Using a martingale argument, they show that
the proportion of mis-classified units among the detected ones
is controlled below a pre-specified level.
Also related are~\cite{mary2021semi} and~\cite{roquain2022false} that
consider testing whether the test data follow a ``null'' distribution
and provide similar analysis.

Without the
subgroup fairness aspect,
the procedure in~\cite{rava2021burden}
is equivalent to our $\cfbhnull$ 
if we set $\hat{V}(x,0)=\hat{s}(x)$.
Our analysis for Proposition~\ref{prop:variant}
is an alternative to the martingale argument of~\cite{rava2021burden,mary2021semi,roquain2022false}.
We  note that both $\cfbh$ and $\cfbhnull$  can be extended to account for subgroup fairness
by applying the selection procedure separately to different subgroups.

While providing a stronger guarantee
(i.e., conditional on the hypotheses),
$\cfbhnull$ 
is less general than $\cfbh$.
$\cfbh$  applies to
problems that cannot be easily translated
to a classfication problem (see, e.g.,
Examples~\ref{ex:missing} and~\ref{ex:counterfactual}).
Moreover, in classification problems where $\cfbhnull$  is applicable,
$\cfbh$  can be
more powerful with a
suitable choice of $V$. We discuss this issue in the following.

\begin{remark}[Power comparison in classification problems] \normalfont
Suppose $Y\in\{0,1\}$ and the goal is to
find $Y=1$. We set $V(x,y) = My - \hat{s}(x)$ as discussed
in Section~\ref{subsec:asymp},
where $\hat{s}(X_i)$ is any score
to predict how likely $Y_i=1$ happens conditional on $X_i$ (such as those used in~\cite{rava2021burden}), and $M$ is a sufficiently large
constant. We suppose $M>  2\sup_x|\hat{s}(x)| $,
for which often $M=2$ suffices.
Since $M>2\sup_x|\hat{s}(x)|$,
for those $i\notin \cD_\calib^0$, i.e., $Y_i=1$, we have
$
V(X_i,Y_i) = M - \hat{s}(X_i) > - \hat{s}(X_{n+j}) = \hat{V}_{n+j}.
$
Thus our conformal p-values used in $\cfbh$  reduce to
\$
p_j 
&= \frac{\sum_{i\in \cD_\calib^0} \ind\{V_i < \hat{V}_{n+j}\} + U_j \cdot( 1+ \sum_{i\in \cD_\calib^0} \ind\{V_i = \hat{V}_{n+j}\}) }{n+1} ,
\$
while the p-values in $\cfbhnull$  are
\$
p_j^0 = \frac{\sum_{i\in \cD_\calib^0} \ind\{V_i < \hat{V}_{n+j}\} + U_j \cdot( 1+ \sum_{i\in \cD_\calib^0} \ind\{V_i = \hat{V}_{n+j}\}) }{|\cD_\calib^0|+1}.
\$
That is, in classification problems, such a choice of nonconformity score  leads to
\$
p_j = \frac{|\cD_\calib^0|+1}{n+1} p_j^0 < p_j^0.
\$
With strictly smaller p-values, 
the rejection set of $\cfbh$  is a superset
of that of $\cfbhnull$,
hence achieving strictly higher power.
Also, the coefficient $\frac{|\cD_\calib^0|+1}{n+1}$
implies that the smaller the proportion of $Y_i=0$ samples
is among the calibration data,
the greater the power gain of $\cfbh$  over $\cfbhnull$.
\end{remark}

\subsection{Connection to outlier detection}
\label{app:sub_outlier}

The outlier detection problem in~\cite{bates2021testing}
can also be turned into
the above classification setting.
%
As stated in Lemma~\ref{lem:outlier},~\cite{bates2021testing}
studies a problem where given i.i.d.~calibration data $\{Z_i\}_{i=1}^n$
from an unknown  distribution $\PP$,
one would like to test for outliers in independent
test samples $\{Z_{n+j}\}_{j=1}^m$,
and the hypotheses are $H_j\colon Z_{n+j}\sim \PP$.
There, $Z_{n+j}$ is called an \emph{inlier} if $H_j$ is true,
and an \emph{outlier} otherwise;
the calibration data are all inliers.

To turn the outlier detection problem into our language,
one could view $Z$ as the covariate,
and encode a label $Y=0$ for inliers and $Y=1$ for outliers.
This setup is the same as the preceding subsection,
where one would like to control the average proportion of
inliers in those classified as $Y=1$.
Here, the inlier distribution
is $\PP_{Z\given Y=0}$;
it is then equivalent to
the hypothesis testing problem in~\cite{mary2021semi} and \cite{roquain2022false} as well,
where data associated with null hypotheses are i.i.d.~from a \emph{null distribution}.
The method in~\cite{bates2021testing} for marginal FDR control
is the same as $\cfbhnull$  if we set
$V(x,0)=-\hat{s}(x)$ for the one-class-classifier $\hat{s}(\cdot)$
trained on inliers
for outlier detection in~\cite{bates2021testing}.
Similar to $\cfbhnull$,
only inliers are used for calibration in~\cite{bates2021testing},
and the FDR control is conditional on the test labels.

Besides the differences in
the methodology and the FDR control guarantee,
another distinction between $\cfbh$  and that of~\cite{bates2021testing}
is the assumption on the data distributions.
In~\cite{bates2021testing} (and also~\cite{mary2021semi,roquain2022false}),
it is only assumed that the inliers
in the test samples are i.i.d.~from $\PP_{Z\given Y=0}$,
but the outliers can be arbitrary distributed and are
not necessarily from the same distribution.
In contrast, $\cfbh$  assumes all $(Z_i,Y_i)$ pairs
in the calibration and test data are from an i.i.d.~super-population.
Leveraging this additional structure,
$\cfbh$ can deal with more
general thresholds for continuous responses, see Examples~\ref{ex:missing} and~\ref{ex:counterfactual}, and
achieve higher power in classification problem
by including all observations in the calibration fold as
we discussed in Section~\ref{app:subsec_classify}.
This super-population assumption can be reasonable
in many cases, such as healthcare diagnosis, job hiring, and
drug discovery.

At a high level,
responses of higher values in our setting
can be
roughly viewed as ``outliers''.
However, many applications such as job hiring
and drug discovery may not be easily turned into an outlier detection
problem. The one-class-classification in~\cite{bates2021testing} may also be insufficient
in such applications when information from
both positive and negative training samples
is available. In contrast,
$\cfbh$  is able to use any prediction model
from an independent training process.

\section{Technical proofs}

\subsection{Proof of Lemma~\ref{lem:key}}
\label{app:subsec_keylem}
Recall that
$V_{n+j}=V(X_{n+j},Y_{n+j})$ is  the unobserved score for the $j$-th test sample, and
the oracle p-value is
\#\label{eq:app_orc_pval}
p_j^* = \frac{\sum_{i=1}^n \ind {\{V_i < {V}_{n+j} \}}+  ( 1+ \sum_{i=1}^n \ind{\{V_i =  {V}_{n+j} \}} ) \cdot U_j }{n+1},
\#
where $U_j$ is the same as in $p_j$.

\begin{proof}[Proof of Lemma~\ref{lem:key}]
The first property follows from the monotonicity of $V$.
To be specific, on the event $\{Y_{n+j}\leq c_j\}$,
we have $V_{n+j} = V(X_{n+j},Y_{n+j})\leq V(X_{n+j},c_j) = \hat V_{n+j}$ hence $p_j^*\leq p_j$.

The second property follows from an application of
the PRDS property proved in~\cite{bates2021testing}.
In the following lemma,
we cite the results from~\cite{bates2021testing}.

\begin{lemma}[PRDS property in~\cite{bates2021testing}]
  \label{lem:outlier}
Let $Z_1,\dots,Z_n\iid \PP$ be the calibration data,
and $Z_{n+1},\dots,Z_{n+m}$ be independent test samples.
For each $j=1,\dots,m$, the null hypothesis is $H_j^*\colon Z_{n+j}\sim \PP$.
For any fixed nonconformity score $V\colon \cZ\to \RR$,
we compute $V_i=V(Z_i)$ for $1\leq i\leq m+n$,
and construct p-values $\{p_j^*\}$ as in~\eqref{eq:app_orc_pval}.
If $H_j^*$ is null, i.e., if $Z_{n+j}\sim \PP$, then
the vector $(p_1^*,\dots,p_m^*)$ is PRDS on $p_j^*$.
\end{lemma}

We now construct another set of `oracle p-values'
in the setup of~\cite{bates2021testing} that coincide with our $\{p_j\}$.
Let $j$ be any fixed index.
We keep $p_j^*$ as it is, and
let $Z_{n+\ell}^* = (X_{n+\ell},c_\ell)$ for all $\ell \neq j$.
Without loss of generality we define $Z_{n+j}^*=(X_{n+j},Y_{n+j})$
and $Z_i^* = (X_i,Y_i)$ for $i=1,\dots,n$.
This can be viewed as the $(X_{n+\ell},Y_{n+\ell})$ pair
by setting $Y_{n+\ell}=c_\ell$ a.s.~to be an ``outlier'' for all $\ell\neq j$.
Then we let $\tilde{V}_{i} = V(Z_{i}^*)$ for $1\leq i\leq n+m$, and
construct p-values using the same $\{U_\ell\}$ as in~\eqref{eq:app_orc_pval}:
\#\label{eq:app_orc_pval_2}
\tilde p_\ell^* = \frac{\sum_{i=1}^n \ind {\{ \tilde{V}_i < \tilde{V}_{n+\ell} \}}+  ( 1+ \sum_{i=1}^n \ind{\{\tilde{V}_i =  \tilde{V}_{n+\ell} \}}) \cdot U_\ell }{n+1},\quad \quad \text{for all }\ell=1,\dots,m.
\#
Note that $\tilde{p}_\ell^* = p_\ell$ for $\ell\neq j$
and $\tilde{p}_j^* = p_j^*$,
where $p_\ell$ is our new conformal p-values in~\eqref{eq:pj}.
In the view of Lemma~\ref{lem:outlier},
$Z_1^*,\dots,Z_n^*$ are the i.i.d.~calibration data,
$Z_{n+1}^*,\dots,Z_{n+m}^*$ are independent test samples,
and $Z_{n+j}^*$ follows the same distribution as $Z_1^*,\dots,Z_n^*$.
Hence by Lemma~\ref{lem:outlier}, we know
the vector $(\tilde p_1^*,\dots, \tilde p_m^*)$ is PRDS on $p_j^*$,
which is equivalent to the second property of Lemma~\ref{lem:key}.
This concludes the proof of Lemma~\ref{lem:key}.
\end{proof}

\subsection{Proof of Theorem~\ref{thm:exchangeable}}
\label{app:subsec_bh}
\begin{proof}[Proof of Theorem~\ref{thm:exchangeable}] 
Let $\cR$ be the rejection set, and let $R_j=\ind\{j\in \cR\}$ for $j=1,\dots,m$.
In the BH procedure,
$j\in \cR$ if and only if $p_j \leq q|\cR|/m$.
The FDR can thus be decomposed as
\$
\fdr &= \EE\bigg[  \frac{\sum_{j=1}^m \ind {\{Y_{n+j}\leq c_j\}}R_j}{1\vee \sum_{j=1}^m R_j}   \bigg]
=   \sum_{j=1}^m \sum_{k=1}^m \frac{1}{k} \EE\big [  \ind {\{|\cR|=k\}} \ind {\{Y_{n+j}\leq c_j\}}\ind {\{p_j\leq q k/m\}}    \big].
\$
Let $\cR_{j\to *}$ be the rejection set obtained by setting $p_j$ to $p_j^*$
while keeping others fixed.
By Lemma~\ref{lem:key}, we have $p_j\leq p_j^*$
on the event $\{Y_{n+j}\leq c_j\}$.
In addition, by the property of the BH procedure,
if $j\in \cR$, i.e., if $p_j \geq q|\cR|/m$,
then sending $p_j$ to a smaller value does not change the rejection set.
Thus,
\$
\ind {\{|\cR|=k\}} \ind {\{Y_{n+j}\leq c_j\}}\ind {\{p_j\leq q k/m\}}
&= \ind {\{|\cR_{j\to *}|=k\}} \ind {\{Y_{n+j}\leq c_j\}}\ind {\{p_j\leq q k/m\}}   \\
&\leq \ind {\{|\cR_{j\to *}|=k\}} \ind {\{Y_{n+j}\leq c_j\}}\ind {\{p_j^*\leq q k/m\}}.
\$
Therefore, the FDR is bounded as
\$
\fdr
&\leq \sum_{j=1}^m \sum_{k=1}^m \frac{1}{k} \EE\big [  \ind {\{|\cR_{j\to *}|=k\}} \ind {\{Y_{n+j}\leq c_j\}}\ind {\{p_j^*\leq q k/m\}}    \big] \\
&\leq \sum_{j=1}^m \sum_{k=1}^m \frac{1}{k} \EE\big [  \ind {\{|\cR_{j\to *}|=k\}}  \ind {\{p_j^*\leq q k/m\}}    \big].
\$
By Lemma~\ref{lem:key},
the vector of p-values $(p_1,\dots,p_{j-1},p_j^*,p_{j+1},\dots,p_m)$
is PRDS~\citep{bates2021testing} on $p_j^*$. 
Thus, following standard proofs for the BH($q$) procedure
under PRDS condition~\cite{benjamini2001control}, each term can be controlled as
\$
\sum_{k=1}^m \frac{1}{k} \EE\big [  \ind {\{|\cR_{j\to *}|=k\}} \ind {\{p_j^*\leq q k/m\}}    \big] \leq \frac{q}{m},
\$
which completes the proof of Theorem~\ref{thm:exchangeable}.
\end{proof}

\subsection{Proof of Theorem~\ref{thm:determ}}
\label{app:determ}

\begin{proof}[Proof of Theorem~\ref{thm:determ}]
For notational simplicity, set
$p_j = p_j^\dtm$ \eqref{eq:pj*} in this proof only.
Also define the corresponding deterministic oracle p-values
\#\label{eq:pj*_determ}
p_j^* = \frac{1+\sum_{i=1}^n \ind\{V_{i}< V_{n+j}\}}{n+1},
\#
and let this notation override~\eqref{eq:pj*} in this proof only.

For any $j=1,\dots,m$, define a set of slightly modified p-values
\#\label{eq:pj_ell}
p_\ell^{(j)} = \frac{\sum_{i=1}^n \ind\{V_i < \hat{V}_{n+\ell}\} + \ind\{V_{n+j}<\hat{V}_{n+\ell}\}}{n+1},\quad \forall ~ \ell \neq j.
\#
These p-values are only used in our analysis
(our method cannot use them since they cannot be computed from the observations).
Also define $\cR(a_1,\dots,a_m)\subseteq\{1,\dots,m\}$ as the rejection (indices) set
obtained by the BH procedure, from p-values taking on the values $a_1,\dots,a_m$.

Recall that the output of Algorithm~\ref{alg:bh} is 
$\cR = \cR(p_1,\dots,p_m)$. 
In the sequel, we will compare $\cR$ to
\$
\cR(p_1^{(j)},\dots,p_{j-1}^{(j)},p_j^*,p_{j+1}^{(j)},\dots,p_m^{(j)})
\$
on the event $\{Y_{n+j} \leq c_j, j\in \cR\}$.
First, on this event, since $V$ is monotone, we have
$V_{n+j}=V(X_{n+j},Y_{n+j})\leq V(X_{n+j},c_j)$, whence $p_j^*\leq p_j$.
For the remaining p-values,
since the scores have no ties, we consider two cases:
\begin{enumerate} 
\item[(i)] If $\hat{V}_{n+\ell}>\hat{V}_{n+j}$,
then $\hat{V}_{n+\ell} > V_{n+j}$ since $\hat{V}_{n+j} > V_{n+j}$.
This means 
\$
p_\ell^{(j)}=\frac{1+\sum_{i=1}^n \ind\{V_i < \hat{V}_{n+\ell}\} }{n+1}= p_\ell.
\$
\item[(ii)] If $\hat{V}_{n+\ell}<\hat{V}_{n+j}$,
then $p_\ell \leq p_j$. Since $j\in \cR$,
the BH procedure implies $\ell\in \cR$.
By definition,  we have
\$
p_\ell^{(j)} \leq \frac{1+ \sum_{j=1}^n \ind\{V_i < \hat{V}_{n+\ell}\}}{n+1}
\leq \frac{1+ \sum_{j=1}^n \ind\{V_i < \hat{V}_{n+j}\}}{n+1} = p_j.
\$
\end{enumerate}
To summarize, 
suppose we are to replace $p_j$ by $p_j^*$ 
and $p_\ell$ by $p_\ell^{(j)}$ for all $\ell\neq j$. 
Then on the event $\{Y_{n+j}\leq c_j,j\in \cR\}$,  
such a replacement does not change any 
of those $p_\ell\geq p_j$; also, 
all those $p_\ell\leq p_j$ including $p_j$ itself (they are rejected in $\cR$) 
are still no greater than 
$p_j$ after the replacement. 
Thus, by the step-up nature 
of the BH procedure, 
such a replacement does not change the rejection set, 
meaning that
\$
\cR  &=  \cR( p_1^{(j)},\dots, p_{j-1}^{(j)}, p_j, p_{j+1}^{(j)},\dots, p_m^{(j)} )\\
&= \cR( p_1^{(j)},\dots, p_{j-1}^{(j)}, p_j^*, p_{j+1}^{(j)},\dots, p_m^{(j)} )=: \cR_{j}^* 
\$
on the event $\{Y_{n+j}\leq c_j,j\in \cR\}$.
As in the proof of Theorem~\ref{thm:exchangeable},
a leave-one-out analysis of the FDR then implies
\$
\fdr &= \EE\bigg[  \frac{\sum_{j=1}^m \ind {\{Y_{n+j}\leq c_j\}}R_j}{1\vee \sum_{j=1}^m R_j}   \bigg] \\
&=  \sum_{j=1}^m \sum_{k=1}^m \frac{1}{k} \EE\big [  \ind {\{|\cR|=k\}} \ind {\{Y_{n+j}\leq c_j\}}\ind {\{p_j\leq q k/m,j\in \cR\}}    \big] \\
&\leq \sum_{j=1}^m \sum_{k=1}^m \frac{1}{k} \EE\big [  \ind {\{|\cR_{j}^*|=k\}} \ind {\{Y_{n+j}\leq c_j\}}\ind {\{p_j^*\leq q k/m\}}    \big] \\
&\leq \sum_{j=1}^m \sum_{k=1}^m \frac{1}{k} \EE\big [  \ind {\{|\cR_{j}^*|=k\}} \ind {\{p_j^*\leq q k/m\}}    \big] \\
&= \sum_{j=1}^m \sum_{k=1}^m \frac{1}{k} \EE\big [  \ind {\{|\cR_{j}^*|=k\}} \ind {\{p_j^* \in \cR_j^*\}}    \big];
\$
the second and the last lines use the property of the BH procedure, whereas
the third uses the facts stated just above. 
By the step-up nature of the 
BH procedure, we know that on the event $\{p_j^*\in \cR_j^*\}$,
sending $p_j^*$ to zero does not change the rejection set, i.e., we have
\$
\cR_{j}^* = \cR(p_1^{(j)},\dots,p_{j-1}^{(j)},0,p_{j+1}^{(j)},\dots,p_m^{(j)})=: \cR_{j\to 0}^*.
\$
Thus
\$
\fdr
\leq \sum_{j=1}^m \sum_{k=1}^m \frac{1}{k} \EE\big [  \ind {\{|\cR_{j}^*|=k\}} \ind {\{p_j^* \in \cR_{j\to 0}^*\}}    \big] = \sum_{j=1}^m  \EE\bigg [ \frac{\ind\{p_j^*\leq q|\cR_{j\to 0}^*|/m\}}{1\vee |\cR_{j\to 0}^*|}   \bigg].
\$

Note that by definition, $\{p_\ell^{(j)}\}_{\ell\neq j}$
is invariant after permuting $\{V_i\}_{i=1}^n \cup \{V_{n+j}\}$.
Since $\{V_i\}_{i=1}^n \cup \{V_{n+j}\}$ are exchangeable conditional on
$\{\hat{V}_{n+\ell}\}_{\ell\neq j}$, we know that
the distribution of $\{p_\ell^{(j)}\}_{\ell\neq j}$ is
independent of $\{V_i\}_{i=1}^n \cup \{V_{n+j}\}$ conditional on
the unordered set $[V_1,\dots,V_n,V_{n+j}]$.
Also note that $\cR_{j\to 0}^*$ only depends on $\{p_j^{(\ell)}\}_{\ell\neq j}$,
and $p_j^*$ only depends on $\{V_i\}_{i=1}^n \cup \{V_{n+j}\}$. This implies  that
$\cR_{j\to 0}^*$ is
independent of $p_j^*$ conditional on  the unordered set $[V_1,\dots,V_n,V_{n+j}]$.
The tower property yields
\$
\EE\bigg [ \frac{\ind\{p_j^*\leq q|\cR_{j\to 0}^*|/m\}}{1\vee |\cR_{j\to 0}^*|}   \bigg]
= \EE\Bigg[\EE\bigg [ \frac{\ind\{p_j^*\leq q|\cR_{j\to 0}^*|/m\}}{1\vee |\cR_{j\to 0}^*|} \bigggiven  [V_1,\dots,V_n,V_{n+j}] \bigg]\Bigg].
\$
In addition, since the variables $\{V_i\}_{i=1}^n\cup \{V_{n+j}\}$ are
conditionally exchangeable, they are also marginally exchangeable.
Thus, for any random variable 
$t\in \RR$ that is measurable with respect to 
the unordered set 
$[V_1,\dots,V_n,V_{n+j}]$, we have 
\$
\PP\big(p_j^*\leq t\biggiven [V_1,\dots,V_n,V_{n+j}]\big) \leq t.
\$
This gives
\$
\EE\bigg [ \frac{\ind\{p_j^*\leq q|\cR_{j\to 0}^*|/m\}}{1\vee |\cR_{j\to 0}^*|} \bigggiven  [V_1,\dots,V_n,V_{n+j}]  \bigg]\leq \frac{q}{m}.
\$
Summing over $j\in\{1,\dots,m\}$ concludes the proof. 
\end{proof}

\subsection{Proof of Proposition~\ref{prop:asymptotic}}
\label{app:prop_asymp}

\begin{proof}[Proof of Proposition~\ref{prop:asymptotic}]
  We utilize an equivalent representation of the BH($q$) procedure,
communicated in~\cite{storey2004strong}:
the rejection set
is $\cR = \{j\colon p_j \leq \hat\tau\}$, where
\#\label{eq:BH_hat_tau}
\hat\tau = \sup \bigg\{ t\in [0,1]\colon \frac{mt}{\sum_{j=1}^m \ind\{p_j \leq t\}} \leq q  \bigg\}.
\#
To clarify the dependence on the calibration data, we denote the p-values as
$
p_j = \hat{F}_n( V_{j}^{0}, U_j),
$
where for simplicity we denote $V_j^0 = \hat{V}_{n+j}=V(X_{n+j},0)$, and define
\$
\hat{F}_n(v,u ) = \frac{\sum_{i=1}^n \ind {\{V_i < v \}}+  ( 1 + \sum_{i=1}^n  \ind {\{V_i = v \}})\cdot u }{n+1}
\$
for any $(v,u )\in \RR\times[0,1] $.
We know that $\{(V_j^0,U_j )\colon 1\leq j\leq m\}$
are i.i.d., and independent of $\cD_\calib$. 
We first  define
\#\label{eq:def_F}
F(v,u) = \PP\big(V(X,Y)<v) + \PP\big(V(X,Y)=v\big)\cdot u.
\#
Then by the uniform law of large numbers, we have
\#\label{eq:ulln}
\sup_{v\in \RR, u\in [0,1] } \big|  \hat{F}_n(v,u ) - F(v,u) \big| ~ \asto ~ 0,\quad
\text{as}~ n\to \infty.
\#

We then
repeatly employ the (uniform) strong law of large numbers to show the asymptotic behavior
of the testing procedure.
Based on~\eqref{eq:ulln}, we show the uniform convergence of the criterion
in~\eqref{eq:BH_hat_tau}.

\begin{lemma}\label{lem:tau_conv}
With the same setup as in the proof of
Proposition~\ref{prop:asymptotic}, suppose $\sup_{x\in \cX}w(x) \leq M$ for some constant $M>0$.
Then
\$
\sup_{t\in[0,1]} \bigg| \frac{1}{m} \sum_{j=1}^m \ind\{\hat{F}_n(V_j^0,U_j )\leq t\}
- \PP\big( F(V_j^0, U_j) \leq t\big)  \bigg| ~\asto ~ 0,
\$
as $m,n\to \infty$,
where $\PP$ is taken with respect to $V_j^0 = V(X_{n+j},0)$ and an independent $U_j\sim$ Unif$[0,1]$.
\end{lemma}

\begin{proof}[Proof of Lemma~\ref{lem:tau_conv}]
Let $0=t_0<t_1<\cdots<t_K =1$ be a partition of $[0,1]$.
Then for each $t\in [0,1]$,
there exists some $k$ such that $t_k \leq t < t_{k+1}$, whence
\#
& \frac{1}{m} \sum_{j=1}^m \ind\{\hat{F}_n(V_j^0,U_j)\leq t_{k}\}
- \frac{1}{m} \sum_{j=1}^m \ind\{F(V_j^0,U_j)\leq t_{k+1}\}  \label{eq:asymp_l}\\
&\leq \frac{1}{m} \sum_{j=1}^m \ind\{\hat{F}_n(V_j^0,U_j)\leq t\}
- \frac{1}{m} \sum_{j=1}^m \ind\{F(V_j^0,U_j)\leq t\} \notag  \\
&\leq  \frac{1}{m} \sum_{j=1}^m \ind\{\hat{F}_n(V_j^0,U_j)\leq t_{k+1}\}
- \frac{1}{m} \sum_{j=1}^m \ind\{F(V_j^0,U_j)\leq t_k\}. \notag \\
&\leq  \bigg|  \frac{1}{m} \sum_{j=1}^m \ind\{\hat{F}_n(V_j^0,U_j )\leq t_{k+1}\}
- \frac{1}{m} \sum_{j=1}^m \ind\{F(V_j^0,U_j)\leq t_k\}   \bigg| \notag \\
&\leq \bigg|  \frac{1}{m} \sum_{j=1}^m \ind\{\hat{F}_n(V_j^0,U_j )\leq t_{k+1}\}
- \frac{1}{m} \sum_{j=1}^m \ind\{F(V_j^0,U_j)\leq t_{k+1}\}   \bigg| \notag \\
&\quad + \bigg|  \frac{1}{m} \sum_{j=1}^m \ind\{F(V_j^0,U_j)\leq t_{k+1}\}
- \frac{1}{m} \sum_{j=1}^m \ind\{F(V_j^0,U_j)\leq t_{k}\}   \bigg|.\label{eq_asymp_upper_1}
\#
For any fixed $\delta>0$, we have
\$
& \bigg|  \frac{1}{m} \sum_{j=1}^m \ind\{\hat{F}_n(V_j^0,U_j )\leq t_{k+1}\}
- \frac{1}{m} \sum_{j=1}^m \ind\{F(V_j^0,U_j)\leq t_{k+1}\}   \bigg| \\
& \leq \frac{1}{m} \sum_{j=1}^m \Big( \ind\{\hat{F}_n(V_j^0,U_j )\leq t_{k+1}, F(V_j^0,U_j)> t_{k+1} \} \\
&\qquad +  \ind\{\hat{F}_n(V_j^0,U_j ) > t_{k+1}, F(V_j^0,U_j) \leq t_{k+1} \}  \Big) \\
& \leq \ind\big\{\sup_j \big|\hat{F}_n(V_j^0,U_j ) -F(V_j^0,U_j)\big|\geq \delta \big\}\\
&\qquad + \frac{1}{m} \sum_{j=1}^m   \ind\{\hat{F}_n(V_j^0,U_j )\leq t_{k+1}, t_{k+1}+\delta \geq F(V_j^0,U_j)> t_{k+1}  \}  \\
&\quad
 + \frac{1}{m} \sum_{j=1}^m  \ind\{\hat{F}_n(V_j^0,U_j ) > t_{k+1}, t_{k+1}- \delta < F(V_j^0,U_j) \leq t_{k+1} \} \\
&\leq  \ind\big\{\sup_j \big|\hat{F}_n(V_j^0,U_j ) -F(V_j^0,U_j)\big|\geq \delta \big\}
+ \frac{1}{m} \sum_{j=1}^m   \ind\{ t_{k+1} - \delta \leq F(V_j^0,U_j) \leq t_{k+1} +\delta \} .
\$
Then~\eqref{eq:ulln} implies $\limsup_{n\to \infty} \ind\big\{\sup_j \big|\hat{F}_n(V_j^0,U_j) -F(V_j^0,U_j)\big|\geq \delta \big\} = 0$ almost surely.
Combinining with the decomposition~\eqref{eq_asymp_upper_1}, we have
\$
& \limsup_{n\to \infty} \bigg|  \frac{1}{m} \sum_{j=1}^m \ind\{\hat{F}_n(V_j^0,U_j )\leq t_{k+1}\}
- \frac{1}{m} \sum_{j=1}^m \ind\{F(V_j^0,U_j)\leq t_k\}   \bigg| \\
& \leq  \frac{1}{m} \sum_{j=1}^m   \ind\{ t_{k+1} - \delta \leq F(V_j^0,U_j) \leq t_{k+1} +\delta \}
 +  \frac{1}{m} \sum_{j=1}^m \ind\{t_k < F(V_j^0,U_j)\leq t_{k+1}\}
\$
almost surely.
Again invoking the (uniform) law of large numbers for i.i.d.~random variables $F(V_j^0,U_j)$,
\$
\limsup_{m\to \infty} \sup_k &\bigg|\frac{1}{m} \sum_{j=1}^m   \ind\{ t_{k+1} - \delta \leq F(V_j^0,U_j) \leq t_{k+1} +\delta \}
 - \PP\big(  t_{k+1} - \delta \leq F(V_j^0,U_j) \leq t_{k+1} +\delta  \big) \\
 &\quad +  \frac{1}{m} \sum_{j=1}^m \ind\{t_k < F(V_j^0,U_j)\leq t_{k+1}\} -
 \PP\big( t_k < F(V_j^0,U_j)\leq t_{k+1} \big) \bigg|  = 0
\$
with probability one. We thus have
\#\label{eq:asymp_upp}
& \limsup_{m,n\to \infty} \sup_k \bigg|  \frac{1}{m} \sum_{j=1}^m \ind\{\hat{F}_n(V_j^0,U_j )\leq t_{k+1}\}
- \frac{1}{m} \sum_{j=1}^m \ind\{F(V_j^0,U_j)\leq t_k\}   \bigg| \notag \\
& \leq \sup_{k} \Big|\PP\big(  t_{k+1} - \delta \leq F(V_j^0,U_j) \leq t_{k+1} +\delta  \big) + \PP\big( t_k < F(V_j^0,U_j)\leq t_{k+1} \big) \Big|
\#
for any partition $\{t_k\}$ and $\delta>0$.
On the other hand, we note that since $U_j \sim $ Unif[0,1] is
independent of the observations, $F(V_j^0,U_j)$ are i.i.d.~with continuous distributions.
Letting $\delta \to 0$ and $\{t_k\}$ be fine enough sends the supremum in~\eqref{eq:asymp_upp}
to zero. With similar arguments, we can show a lower bound
for~\eqref{eq:asymp_l} that leads to
\$
\limsup_{m,n\to \infty} \sup_k \bigg| \frac{1}{m} \sum_{j=1}^m \ind\{\hat{F}_n(V_j^0,U_j )\leq t_{k}\}
- \frac{1}{m} \sum_{j=1}^m \ind\{F(V_j^0,U_j)\leq t_{k+1}\} \bigg| = 0
\$
almost surely.
Combining the above two results, we then have
\#\label{eq_asymp_1}
\sup_{t\in[0,1]} \bigg|  \frac{1}{m} \sum_{j=1}^m \ind\{\hat{F}_n(V_j^0,U_j)\leq t\}
- \frac{1}{m} \sum_{j=1}^m \ind\{F(V_j^0,U_j)\leq t\}  \bigg| ~\asto ~ 0.
\#
Invoking the uniform strong law of large numbers, we have
\$
\sup_{t\in[0,1]} \bigg|   \frac{1}{m} \sum_{j=1}^m \ind\{F(V_j^0,U_j)\leq t\} - \PP\big( F(V_j^0, U_j) \leq t\big)  \bigg| ~\asto ~ 0,
\$
hence by the triangular inequality we complete the proof of Lemma~\ref{lem:tau_conv}.
\end{proof}

With similar arguments as in the proof of Lemma~\ref{lem:tau_conv},
we can also show that as $n\to \infty$,
\#\label{eq:unif_h0}
\sup_{t\in[0,1]} \bigg|   \frac{1}{m} \sum_{j=1}^m \ind\{\hat{F}_n(V_j^0,U_j)\leq t, j\in \cH_0\} - \PP\big( F(V_j^0, U_j) \leq t, Y(1)\leq Y(0)\big)  \bigg| ~\asto ~ 0.
\#

Suppose there exists some $t' \in (0,1]$ such that
$
\frac{\PP ( F(V_j^0, U_j) \leq t')}{t'} > \frac{1}{q}.
$
We then define
\$
t^* = \sup\bigg\{ t\in [0,1]\colon  \frac{\PP\big( F(V_j^0, U_j) \leq t \big)}{t }   \geq \frac{1}{q}    \bigg\}
 = \sup\big\{  t\in[0,1]\colon G_\infty(t) \leq q \big\},
\$
where $G_\infty(t) = t/\PP(F(V_j^0,U_j)\leq t)$.
It is well-defined and $t^*\geq t'$.
Fix any $\delta\in (0, t')$. By Lemma~\ref{lem:tau_conv},
\#\label{eq:unif_conv_tau}
\sup_{t\in[\delta,1]} \bigg|   \frac{\sum_{j=1}^m \ind\{F(V_j^0,U_j)\leq t\}}{mt}  - \frac{\PP  ( F(V_j^0, U_j) \leq t )}{t}  \bigg| ~\asto ~ 0.
\#
In particular, $\frac{\sum_{j=1}^m \ind\{F(V_j^0,U_j)\leq t'\}}{mt'} ~\asto~ \frac{\PP ( F(V_j^0, U_j) \leq t' )}{t'} > \frac{1}{q}$,
hence $\hat\tau \geq t' \geq \delta$ eventually. 
Furthermore, since $F(V_j^0,U_j)$ admits a continuous distribution,
the function $t\mapsto \PP\big( F(V_j^0, U_j) \leq t \big)$ is continuous in $t\in [0,1]$.
Under the assumption that for any $\epsilon>0$, there exists
some $|t-t^*|\leq \epsilon$ such that $\frac{\PP ( F(V_j^0, U_j) \leq t)}{t}>1/q$,
the uniform convergence in~\eqref{eq:unif_conv_tau} implies $\hat\tau ~\asto~ t^*$.

Let $\delta \in (0, t^*)$ be any fixed value
such that $\PP ( F(V_j^0, U_j) \leq \delta )>2\epsilon$
for some constant $\epsilon>0$.
Then we know $\inf_{t\in[\delta,1]}\PP\{F(V(X,Y(0)),U)\leq t \}\geq \epsilon$,
and by Lemma~\ref{lem:tau_conv},
\$
\liminf_{m\to\infty} \inf_{t\in[\delta,1]}\bigg\{\frac{1}{m}\sum_{j=1}^m \ind\{\hat{F}_n(V_j^0, U_j ) \leq t\}  \bigg\}\geq \epsilon
\$
almost surely.
Combining this lower boundedness property
with the uniform convergence results in
Lemma~\ref{lem:tau_conv}, equation~\eqref{eq_asymp_1},
and~\eqref{eq:unif_h0}, we know that
\$
\sup_{t\in [\delta,1]} \Bigg|\frac{\sum_{j=1}^m  \ind\{\hat{F}_n(V_j^0, U_j) \leq t, j\in \cH_0\}}{1\vee \sum_{j=1}^m \ind\{\hat{F}_n(V_j^0, U_j ) \leq t\}} - \frac{ \PP \{ F(V(X,Y(0)),U)\leq t, Y(1)\leq Y(0) \}}{\PP\{F(V(X,Y(0)),U)\leq t \} }\Bigg| ~\asto~ 0.
\$
Since $\hat\tau ~\asto ~  t^*$ and the distribution functions are continuous, the asymptotic FDR is
\$
\lim_{m,n\to\infty} \fdr &
= \lim_{m,n\to \infty} \EE\Bigg[\frac{\sum_{j=1}^m  \ind\{\hat{F}_n(V_j^0, U_j) \leq \hat\tau,
Y_{n+j}\leq 0 \}}{1\vee \sum_{j=1}^m \ind\{\hat{F}_n(V_j^0, U_j) \leq \hat\tau\}} \Bigg] \\
&=  \EE\Bigg[\lim_{m,n\to \infty}\frac{\sum_{j=1}^m  \ind\{\hat{F}_n(V_j^0, U_j ) \leq \hat\tau, Y_{n+j}\leq 0 \}}{1\vee \sum_{j=1}^m \ind\{\hat{F}_n(V_j^0, U_j ) \leq \hat\tau\}} \Bigg] \\
&= \frac{ \PP \{ F(V(X,0),U)\leq t^*, Y \leq 0  \}}{\PP\{F(V(X,0),U)\leq t^* \} },
\$
where the second line follows from the Dominated Convergence Theorem.
With similar arguments, we can show that the asymptotic power of the procedure is
\$
\lim_{m,n\to\infty}  \text{Power}  &=
\lim_{m,n\to \infty} \EE\Bigg[\frac{\sum_{j=1}^m  \ind\{\hat{F}_n(V_j^0, U_j )   \leq \hat\tau, Y_{n+j}>0\}}{1\vee \sum_{j=1}^m \ind\{Y_{n+j} >0 \}} \Bigg] \\
&= \frac{ \PP \{ F(V(X,0),U)\leq t^*, Y>0 \}}{\PP\{Y>0 \} }.
\$
Therefore, we complete the proof of Proposition~\ref{prop:asymptotic}.
\end{proof}

\subsection{Proof of Proposition~\ref{prop:variant}}
\label{app:subsec_variant}
\begin{proof}[Proof of Proposition~\ref{prop:variant}]
We show that the FDR conditional on all signs of the data is controlled, i.e.,
\$
\EE\Bigg[  \frac{\sum_{j=1}^m \ind\{Y_{n+j}\leq 0,j\in \cR\} }{1\vee \sum_{j=1}^m R_j} \bigggiven \ind {\{Y_i\leq 0\}}\colon i\in \cD_\calib^- \cup \cD_\test  \Bigg] \leq q.
\$
Following the same arguments as in the proof of Theorem~\ref{thm:exchangeable}, it suffices to show
\$
\sum_{k=1}^m \frac{1}{k} \EE\big [  \ind {\{|\cR_{j\to *}|=k\}} \ind {\{p_j^*\leq q k/m\}} \ind {\{Y_{n+j}\leq 0  \}}   \biggiven \ind {\{Y_i\leq 0\}}\colon i\in \cD_\calib^- \cup \cD_\test  \big] \leq \frac{q}{m},
\$
where $\cR_{j\to *}$ is the rejection set obtained by changing $p_j$ to its oracle counterpart
\$
p_j^* = \frac{\sum_{i\in \cD_\calib^-} \ind {\{V_i < {V}_{n+j} \}}+   ( 1+ \sum_{i\in \cD_\calib^-} \ind {\{V_i =  {V}_{n+j} \}}) \cdot U_j }{n+1}.
\$
Recall that we define $Z_i=(X_i,Y_i)$ for $1\leq i\leq n+m$
and $\tilde{Z}_{n+j} = (X_{n+j},c_j) = (X_{n+j}, 0)$ for $j=1,\dots,m$.
Conditional on all signs,
for any fixed $j$ with $Y_{n+j}\leq 0$,
$\{Z_i\}_{i=1}^n \cup \{\tilde{Z}_{n+\ell}\}_{\ell\neq j} \cup \{Z_{n+j}\}$
are mutually independent, and
$\{Z_i\colon i=1,\dots,n,n+j\}$ are
i.i.d.
The desired result thus follows from the PRDS of conformal p-values,
or equivalently the conditions in Theorem~\ref{thm:exchangeable}.
\end{proof}

\section{Additional  details and results}

\subsection{Detailed setup for the data illustration in Section~\ref{subsec:motivate}}
\label{app:subsec_motivate}

We use the same dataset, same data splitting scheme, and 
the same machine learning model $\hat\mu(\cdot)$ as in Section~\ref{subsubsec:hiv}. 
That is, we repeat the procedure independently for $N=100$ times; 
in each time, we randomly split the whole dataset into training $\cD_\train$, calibration $\cD_\calib$, 
and test data $\cD_\test$ with ratio $6\colon 2 \colon 2$ in size. The model $\hat\mu(\cdot)$ is 
trained on $\cD_\train$; $\cD_\calib$ is used to construct conformal prediction sets; 
the prediction sets on $\cD_\test$ are used to evaluate FDP, miscoverage rate, and proportion 
of $\hat{C}_{1-\alpha}(X)=\{1\}$. 
These quantities are then averaged over $N=100$ runs to estimate the FDR, marginal miscoverage, 
and average proportion of $\hat{C}_{1-\alpha}(X)=\{1\}$. 
In constructing conformal prediction sets, we set $V(x,y)=y-\hat\mu(x)$ 
for all three methods. 
The split conformal prediction set with $(1-\alpha)$ marginal coverage is 
\$
\hat{C}_{1-\alpha}(x) = \{ y\in\{0,1\}\colon V(x,y) \geq \hat\eta\},
\$
where $\hat\eta$ is the $1-(1-\alpha)(1+1/n_{\calib})$-th empirical quantile 
of $\{V(X_{i},Y_i)\colon i\in \cD_\calib\}$. 
The naive approach takes $\cR= \{j \in \cD_\test \colon \hat{C}_{1-\alpha}(X_{j}) = \{1\}\}$. 
Our approach is Algorithm~\ref{alg:bh}.  
Bonferroni correction sets $\cR= \{j \in \cD_\test \colon p_j \leq q/m\}$, 
where $\{p_j\}$ are the conformal p-values constructed in our approach. 
When evaluating our approach and Bonferroni correction (i.e., producing the plot 
on the right panel), we randomly take a subset of $m=1000$ test samples, 
such that conformal prediction at resolution $q/m$ is still feasible. 

\subsection{Data generating processes}
\label{app:subsec_simu_detail}

To better illustrate the data distributions,
Figure~\ref{fig:simu_sample} shows the scatterplots
of data from our eight settings.

\begin{figure}[h]
    \centering
    \includegraphics[width=\linewidth]{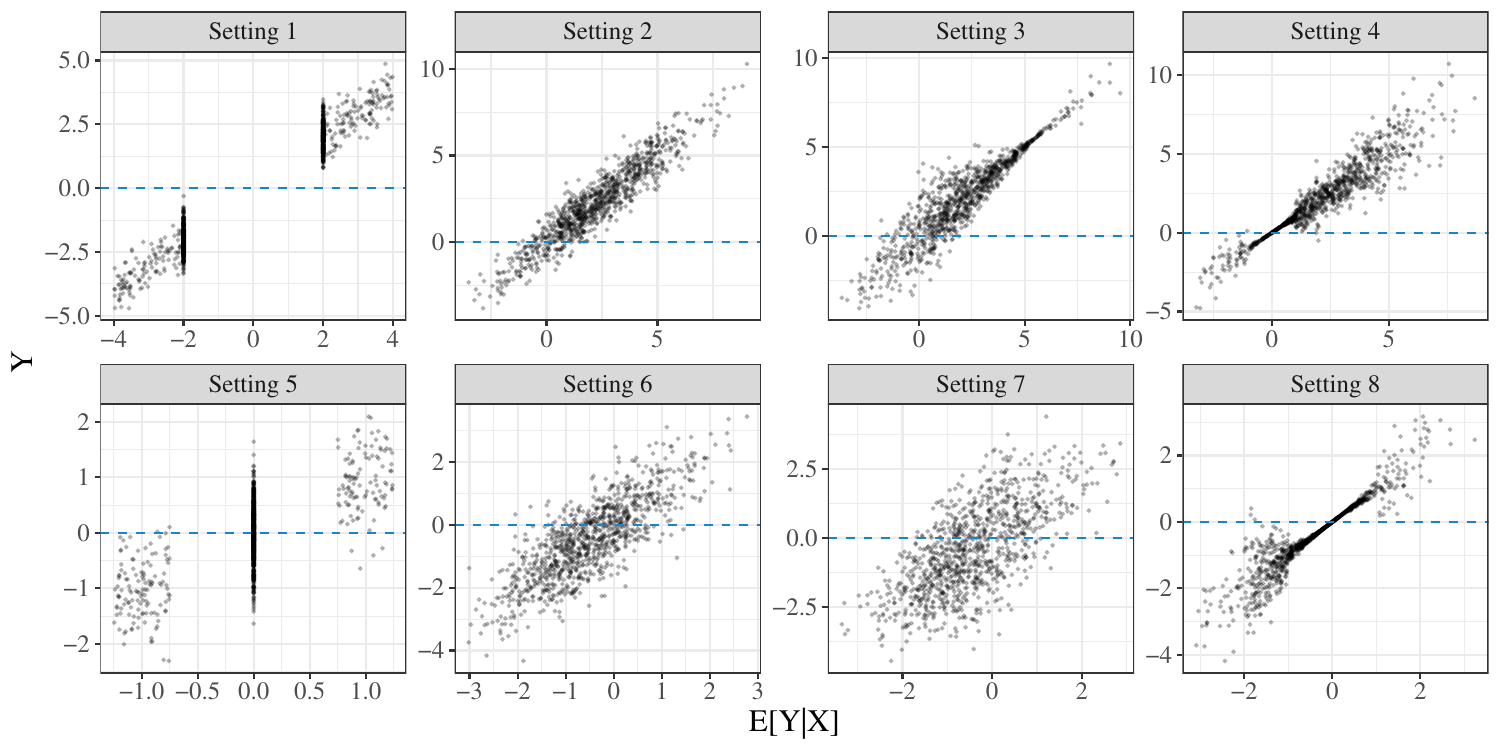}
    \caption{Scatter plots of i.i.d.~samples $\{(X_i,Y_i)\}_{i=1}^{1000}$
    from the data generating processes in our simulations.
    The $x$-axis is the conditional mean function $\mu(X_i)$; the $y$-axis is the
    actual outcome $Y_i$.}
    \label{fig:simu_sample}
  \end{figure}

  In particular,
 we vary the following two aspects
  such that
  the hardness of correctly identifying those $Y>0$
  is not the same.

  \begin{enumerate}
    \item Continuity of the range of $\mu$:
    settings 1 and 5 have disjoint ranges of $\mu(\cX)$ for negative, zero, and positive mean outcomes.
    Among the remainings with continuous ranges of $\mu$, settings 2, 3, 4
    have more samples with positive $\mu(x)$, and
    settings 6, 7, 8 have more samples with negative $\mu(x)$.
    While directional selection seems to be easy
    in setting 1 and 5 (where
    a perfect selection would be to choose those $\mu(X_i)>0$),
    with large noise level it may still be diffcult
    to learn the true model.
    The hardness of prediction
    in continuous-range settings also depends on the absolute scale of the mean functions.

    \item Noise heterogeneity: in settings 1, 2, and 5, 6,
    we set $\epsilon_i \sim N(0,\sigma^2)$ with homogenous variance.
    Other continuous-range settings
    all have heterogenous noise $\epsilon_i\given X_i\sim N(0,\sigma(X_i)^2)$
    for some function $\sigma(\cdot)$.
    To be specific, in settings 3, 4, 8, the variance of noise
    increases with $|\mu(x)|$, showing more difficulty in
    the direction of interest.
    In setting 7, the noise is smaller for
    larger $|\mu(x)|$.
    In general, the task would be easier if
    the mean value is large and the noise variance is small for
    the direction of interest. 
  \end{enumerate}

  The specific configurations of $\mu(x)$ and $\sigma(x)$
  in all of our simulation settings are detailed in Table~\ref{tab:simu} to
  reproduce the results in Section~\ref{sec:simu}.

\begin{table}[h]
\centering
\renewcommand\arraystretch{1.5}
\begin{tabular}{c|c |c }
\hline
  Setting  & $\mu(\cdot)$ & $\sigma(\cdot)$ for $\epsilon_i\given X_i=x \sim N(0,\sigma(x)^2)$ \\
\hline
\multirow{2}{*}{1}  &   $4x_1\ind\{x_2>0\} \cdot \max\{0.5,x_3\}  $ &    \multirow{2}{*}{$ \sigma^2 $}    \\
 & $\quad \quad + 4x_1 \ind\{x_2\leq 0\} \cdot \min\{x_3,-0.5\}$ &  \\
 \hline
2  & $5(x_1x_2+e^{x_4-1})$  &  $ 2.25\sigma^2 $   \\
\hline
3  &  $5(x_1x_2+e^{x_4-1})$ &   $ \sigma\cdot (5.5-|\mu(x)|)/2$   \\
\hline
\multirow{2}{*}{4}  &  \multirow{2}{*}{$5(x_1x_2+e^{x_4-1})$}  &  $\sigma\cdot 0.25 \mu(x)^2 \ind\{|\mu(x)|<2\} $   \\
 & & $\quad\quad+ \sigma\cdot 0.5 |\mu(x)| \ind\{|\mu(x)|\geq 1\}$\\
\hline
\multirow{2}{*}{5}   & $x_1\ind\{x_2>0,x_4>0.5\} \cdot(0.25+x_4)$  &  \multirow{2}{*}{$\sigma$}    \\
 & $\quad\quad  + x_1\ind\{x_2\leq 0,x_4<-0.5\} \cdot (x_4-0.25)$ & \\
\hline
6  &  $2(x_1x_2+x_3^2+e^{x_4-1}-1)$ &  $1.5\sigma$   \\
\hline
7  &  $2(x_1x_2+x_3^2+e^{x_4-1}-1)$ &  $\sigma\cdot(5.5-|\mu(x)|)/2$   \\
\hline
\multirow{2}{*}{8}  & \multirow{2}{*}{$2(x_1x_2+x_3^2+e^{x_4-1}-1)$}  &   $\sigma\cdot(0.25\mu(x)^2 \ind\{|\mu(x)|<2\} $  \\
&&$\quad\quad \quad + 0.5|\mu(x)|\ind\{|\mu(x)|\geq 1\})$\\
\hline
\end{tabular}
\vspace{0.5em}
\caption{Details of the eight data generating processes in the simulations of Section~\ref{sec:simu}. The parameter $\sigma$  corresponds to the noise strength on
the $x$-axis in Figures~\ref{fig:fdr_01}
and~\ref{fig:power_01}. }
\label{tab:simu}
\end{table}

\subsection{Additional plots}
\label{app:subsec_plots}

\begin{figure}[htbp]
  \centering
  \includegraphics[width=\linewidth]{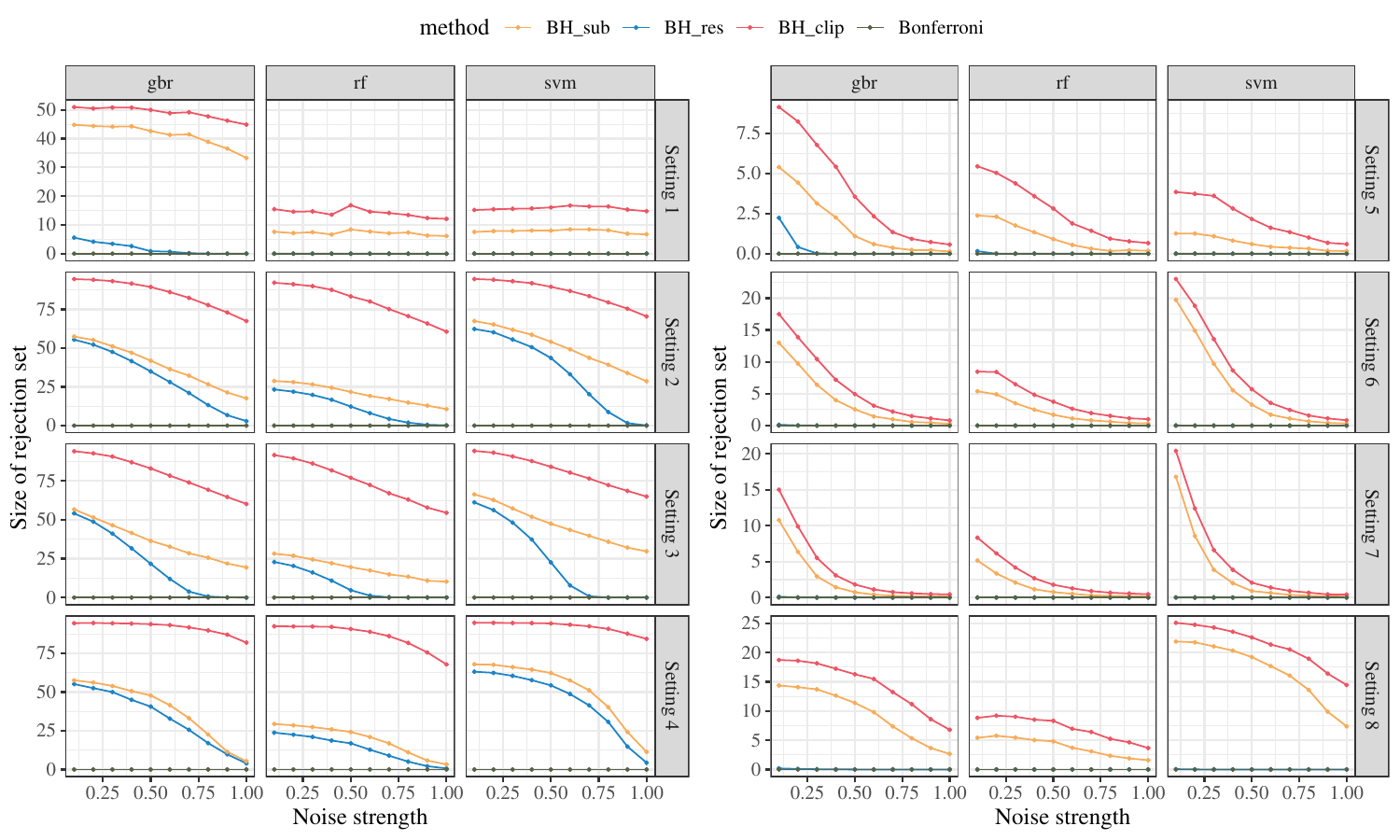}
  \caption{Average size of rejection set for four procedures at
  FDR target $q=0.1$ for various data generating processes.
  Details of the plots are otherwise the same as Figure~\ref{fig:fdr_01}.}
  \label{fig:nrej_01}
\end{figure}

\end{document}